\documentclass{llncs}
\usepackage[utf8]{inputenc}
\usepackage[T1]{fontenc}
\usepackage[english]{babel}
\usepackage{parskip}

\usepackage{amsmath}
\usepackage{amssymb}
\usepackage{amsfonts}
\usepackage{stmaryrd}
\usepackage{xcolor}
\usepackage{algorithm}
\usepackage{algpseudocode}
\usepackage{xspace}
\usepackage{mathtools}
\usepackage{enumerate}
\usepackage{hyperref}

\usepackage{macros}

\usepackage{float}
\usepackage{placeins}
\usepackage{caption} 
\usepackage{booktabs}
\usepackage{multirow}

\usepackage[adversary,operators,mm,primitives,keys,complexity,advantage,notions]{cryptocode}
\createpseudocodeblock{pcb}{center, boxed }{}{}{}

\title{MinRank Gabidulin encryption scheme on matrix codes}

\author{Nicolas Aragon\inst{1} \and Alain Couvreur\inst{2} \and Victor Dyseryn\inst{3} \and Philippe Gaborit\inst{1} \and Adrien Vinçotte\inst{1}}

\institute{
XLIM, Université de Limoges, France \and LIX, École Polytechnique, France \and Télécom Paris, France}

\begin{document}
	
\maketitle

\begin{abstract}
  The McEliece scheme is a generic frame introduced in \cite{M78},
  which allows to use any error correcting code for which there exists
  an efficient decoding algorithm to design an encryption scheme by
  hiding the generator matrix of the code. Similarly, the Niederreiter
  frame, introduced in \cite{N86}, is the dual version of the McEliece
  scheme, and achieves smaller ciphertexts.  In the present paper, we
  propose a generalization of the McEliece frame and the Niederreiter
  frame to matrix codes and the MinRank problem, that we apply to
  Gabidulin matrix codes (Gabidulin rank codes considered as matrix
  codes). The masking we consider consists in starting from a rank
  code $\mathcal C$, computing a matrix version of $\mathcal C$ and then
  concatenating a certain number of rows and columns to the matrix
  code version of the rank code $\mathcal C$ before applying an isometry for
  matrix codes, i.e. right and left multiplications by fixed random
  matrices. The security of the schemes relies on the MinRank problem
  to decrypt a ciphertext, and the structural security of the scheme
  relies on the new EGMC-Indistinguishability problem that we
  introduce and that we study in detail. The main structural attack
  that we propose consists in trying to recover the masked linearity
  over the extension field which is lost during the masking
  process. Overall, starting from Gabidulin codes we obtain a very
  appealing tradeoff between the size of the ciphertext and the size of
  the public key. For 128 bits of security we propose parameters
  ranging from ciphertexts of size 65 B (and public keys of size 98
  kB) to ciphertexts of size 138B (and public keys of size 41 kB).
  For 256 bits of security, we can obtain ciphertext as low as 119B,
  or public key as low as 87kB.  Our new approach permits to achieve a
  better trade-off between ciphertexts and public key than the
  classical McEliece scheme instantiated with Goppa codes.
  
\end{abstract}

\section{Introduction}

\subsubsection*{Matrix codes and vector codes}

Introduced in 1951 in \cite{H51}, before the well-known notion of Hamming metric, matrix codes are $\Fq$-linear subspaces of the matrix space $\Fq^{m \times n}$, endowed with the rank metric.
The generic decoding problem in matrix codes relies on the difficulty to solving an instance of the MinRank problem, which is known to be NP-complete. A rank metric code of length $n$ over the extension field $\Fqm$ is an $\Fqm$--linear vector subspace of $\Fqm^n$, endowed with the rank metric. It is possible to turn a vector code into a matrix code by considering all vectors of a vector code
as matrices over $\Fq$ by writing of the extension field $\Fqm$ as a vector space of dimension $m$ over $\Fq$. The main advantage of vector codes compared to matrix codes lies in the fact that the linearity over the extension field permits to achieve smaller description of the generator matrix of the code (division by a factor $m$, the degree of the field extension).

\subsubsection*{Code-based cryptography.}
There exists two approaches for code-based cryptography: the first one in which the public key consists in a masking of a structured code, the McEliece frame,
for instance using Goppa codes \cite{M78} or MDPC codes \cite{MTSB12}.
In that case the hidden code is used both for encryption and decryption. The second one is the  Aleknovich approach \cite{A03} and its variants like \cite{AABBBDGZ17,AABBBDGPZ21a} in which there is no structural masking and for which two types of codes are considered, a random (or random quasi-cyclic) code for encryption and another code for decryption.
The main advantage of the second approach is that it avoids structural attacks, so it is semantically stronger in terms of security. However, the price to pay is a larger ciphertext (typically quadratic in the security parameter), whereas the first approach permits to obtain smaller ciphertexts (linear in the security parameter) at a cost of a very large public key.
For instance for 128 bits of security for McEliece it is possible to get ciphertexts of size of order 100B, whereas for 
HQC the ciphertext is of order 1500kB.
There exist applications for which having a very small ciphertext may be of interest.

\subsubsection*{Structural attacks and distinguishers on McEliece-like schemes.}
 The main tool for structural attacks is to consider a distinguisher for the underlying masked code. While it is really interesting to consider Reed-Solomon codes or Gabidulin codes -- their rank metric analog -- for encryption because of their very good decoding properties, these codes are very highly structured which makes them easy to distinguish from random codes. For Reed-Solomon codes, the main distinguisher is the square code distinguisher of \cite{CGGOT13} which states that the square code of an $[n, k]$ Reed-Solomon code is a code of dimension $2k-1$, whereas for a random code its dimension is close to $\min (n, \frac{k(k-1)}{2})$. There exists an analog distinguisher for Gabidulin codes \cite{O08}, starting from an $[m,k]$ Gabidulin code over $\Fqm$ with generator matrix $G$, and considering $G^q$ the matrix in which one applies the Frobenius action $x \mapsto x^q$
to all entries of $G$. The dimension of the code whose generator matrix is a vertical concatenation of $G$ and $G^q$ has dimension $k+1$, whereas for a random code it would be $2k$. These two distinguishers were used to break systems in which the masking consisted in adding random columns (and then applying an isometry) to Reed-Solomon codes \cite{Wie06} in \cite{CGGOT13} or to Gabidulin codes \cite{OG23} in \cite{O08}.

\subsubsection*{\bf The case of the original Goppa-McEliece scheme.}
The main masking which is still resistant to structural attacks is the case of the original Goppa-McEliece scheme.
This case can be seen as considering generalized Reed-Solomon codes over an extension field $\mathbb{F}_{2^m}$, then considering its $\mathbb{F}_{2}$ subcodes.
The effect of considering the $\mathbb{F}_{2}$ subcode is that it breaks the structure over the $\mathbb{F}_{2^m}$ extension and hence the direct application of a Reed-Solomon distinguisher, at least in the case were the rate of the code is not close to 1, which is the general case (see also \cite{FGOPT10}).

\subsubsection*{\bf Contributions.} Inspired by the previous situation we consider the case of Gabidulin codes in which we want to break the structure over the extension on which relies the Gabidulin distinguisher. Indeed, in order to apply the Overbeck distinguisher we need to consider the application of the Frobenius map which only makes sense over an extension $\Fqm$ and not directly on the base field $\mathbb{F}_{q}$.
We consider the following masking: we start from a matrix code, which can decode up to errors of rank $r$, we then add random rows and columns and apply an isometry for matrix codes (multiplying on the left and on the right by two fixed random invertible matrices). Starting from a vector code $\mathcal C$, we call \emph{enhanced code} the transformation which consists in turning the vector code $\mathcal C$ into a matrix code and then applying the previous masking.
Since it is a masking, there always exists the possibility of a structural attack but the idea is that considering matrix version of the vector codes with our masking (in the spirit of the original Goppa-McEliece scheme) permits to avoid an efficient action of a distinguisher.

The main contributions of the paper are the following: \begin{itemize}
\item we describe a general encryption McEliece-like frame for matrix codes over the MinRank problem,
\item we propose a general masking for matrix codes that we apply on a matrix code version of Gabidulin codes,
\item we study in detail possible distinguishers for solving the Enhanced Gabidulin Matrix Code (EGMC) distinguishing problem that we introduce for our scheme.
\end{itemize}

In terms of distinguisher the best attack that we obtain for the
vector code we consider consists in trying to dismiss the action of
the addition of random rows and columns on the matrix code in order to
test for its underlying linearity over the extension field.  In terms
of parameters, our new approach permits to obtain an alternative
scheme to the classic McEliece scheme with very small ciphertexts and
even smaller public keys than in the classic McEliece scheme.  For 128
bits of security, we can achieve a size of 65B for the ciphertext and
98kB for the public key, or 138B for the ciphertext and 41kB for the
public key. These results compare very well to classic McEliece's
parameters \cite{ABCCGLMMMNPPPSSSTW20}) yielding ciphertext of 96 B
and public keys of 261 kB.
For 256 bits of security, we can obtain ciphertext as low as 119B, or public key as low as 87kB.

\subsubsection*{Organization of the paper.} The paper is organized as follows. Section~\ref{sec:prelim} gives an introduction and preliminaries, Section~\ref{sec:emc} defines the masking transformations we introduce and the attached distinguishing problems, Sections~\ref{new_frames} and~\ref{nes} describe the new encryption schemes we propose, Section~\ref{sec:secu} examines in detail different types of possible distinguishers and Section~\ref{sec:param} details parameters of our schemes.

\section{Preliminaries}\label{sec:prelim}

\subsection{Public key encryption}

\begin{definition}[PKE]
A public key encryption scheme PKE consists of three polynomial time algorithms:\begin{itemize}
    \item $\KeyGen (1^{\lambda})$: for a security parameter $\lambda$, generates a public key $\pk$ and a secret key $\sk$.
    \item $\Encrypt (\pk, \mv)$: outputs a ciphertext $\cv$ given the message $\mv$ and the public key $\pk$.
    \item $\Decrypt (\sk, \cv)$: outputs the plaintext $\mv$ of the encrypted ciphertext $\cv$, or $\perp$.
\end{itemize}
We require a PKE scheme to be correct: for every pair $(\pk, \sk)$ generated by $\KeyGen$ and message $\mv$, we should have $\Pr \left(\Decrypt(\sk,\Encrypt(\pk,\mv))=\mv\right)=1 - \negl (\lambda)$, where $\negl$ is a negligible function.
\end{definition}

There exist several notions of security for PKE schemes. The encryption schemes we propose here achieve OW-CPA security.

\begin{definition}[OW-CPA security]
    Let $\PKE = (\KeyGen, \Encrypt, \Decrypt)$ a PKE scheme, $\mA$ an adversary against $\PKE$, and $\lambda$ a level of security. We define the OW-CPA game: \begin{itemize}
        \item \textbf{Challenge.} The challenger generates $(\pk,\sk)\leftarrow\KeyGen (1^\lambda)$, samples $\mv$ from the set of messages and computes $\cv\leftarrow\Encrypt (\pk,\mv)$. He sends $(\cv,\pk)$ to the adversary $\mA$.
        \item \textbf{Output.} $\mA$ outputs the guessing message $\Tilde{\mv}$. $\mA$ wins if $\Tilde{\mv}=\mv$.
    \end{itemize}
$\PKE$ is OW-CPA secure if for any PPT adversary $\mA$, the probability that $\mA$ wins the game is negligible in $\lambda$.
\end{definition}

To prove the security of our PKE, we need to formally define the notion of indistinguishable distributions.

\begin{definition}[Distinguisher]
    Let $S$ a set of elements, $\mathcal{P}_1$ and $\mathcal{P}_2$ two probability distributions on $S$. A distinguisher $\mD$ for the distributions $\mathcal{P}_1$ and $\mathcal{P}_2$ is an algorithm that takes in input an element of $S$, and outputs a bit. Its advantage is defined by: $$\Adv (\mD) = \big\lvert \Pr (\mD (x)=1\vert x\leftarrow\mathcal{P}_1) - \Pr (\mD (x)=1\vert x\leftarrow\mathcal{P}_2)\big\rvert.$$
\end{definition}

We say that the distributions $\mathcal{P}_1$ and $\mathcal{P}_2$ are indistinguishable if there exists no polynomial time distinguisher with non negligible advantage.

\subsection{Matrix codes}

\begin{definition}[$\gamma$-expansion]\label{g-e}
  Let $\gamma = (\gamma_1,\dots,\gamma_m)$ be an $\Fq$-basis of $\Fqm$.
  The $\gamma$-expansion of an element in $\Fqm$ to a vector in $\Fq^m$ is defined as the application:
\[
\Psi_\gamma: x \in \Fqm \mapsto (x_1,\dots,x_m) \in \Fq^m 
\]
such that $x = \sum_{i=1}^{m} x_i \gamma_i$.
\end{definition}

$\Psi_\gamma$ extends naturally to a word $\xv \in \Fqm^n$ and turns
it into a matrix $\Psi_\gamma(\xv)\in\Fqmn$, by writing in columns the
coordinates of each element in basis $\gamma$ in column. We denote as
$\mathcal{B}(\Fqm)$ the set of $\Fq$-bases of $\Fqm$.

\begin{definition}[Vector codes with rank metric]
A vector code $\Cvec$ is an $\Fqm$-subspace of $\Fqm^n$ endowed with the rank metric. The weight of a vector $\xv\in\Fqm^n$ is the rank of the matrix $\Psi_\gamma(\xv)$, for a basis $\gamma\in\mathcal{B}(\Fqm)$. The weight of a vector is independent of the choice of the basis $\gamma$.
\end{definition}

Equivalently, the \emph{rank of a vector $\xv$} is the dimension of
the $\Fq$-vector space spanned by its coordinates, that is the \emph{support}
of $\xv$.
We denote the weight of a vector $\xv$ by:
\[
    \wrank{\xv} \eqdef \rank (\Psi_\gamma(\xv))
    =\dim (\langle x_1,...,x_n\rangle_q)
\]

\begin{definition}[Matrix codes]
A matrix code $\Cmat$ is an $\Fq$-subspace of $\Fqmn$ endowed with the rank metric.
\end{definition}

An $\Fqm$-linear vector code $\Cvec$ of parameters $[n,k]_{q^m}$ can
be turned into a matrix code of size $m\times n$ and dimension $mk$,
defined as:
\[
\Psi_\gamma(\Cvec) := \{ \Psi_\gamma(\xv) \,|\, \xv \in \Cvec\}.
\]
Let $\gamma\in\mathcal{B}(\Fqm)$, and $(\vv_i)_{i\in\{1,...,k\}}$ be an $\Fqm$-basis of $\Cvec$, then an $\Fq$-basis of $\Psi_{\gamma}(\Cvec)$ is given by:
$$\Big\lbrace \Psi_{\gamma}(b\vv_i)~\vert~ b\in\gamma, i\in\{1,...,k\}\Big\rbrace.$$

\begin{definition}[Equivalent matrix codes]
    Two matrix codes $\Cmat$ and $\Dmat$ are said to be \emph{equivalent} if there exist two matrices $\Pv \in \mathbf{GL}_{m}(\Fq)$ and $\Qv \in \mathbf{GL}_{n}(\Fq)$ such that $\Dmat=\Pv\Cmat\Qv$. If $\Pv = \Iv_m$ (resp. $\Qv = \Iv_n$), $\Cmat$ and $\Dmat$ are said to be right equivalent (resp. left equivalent).
\end{definition}

As seen above, the transformation of a vector code into a matrix code is dependent on the choice of the basis $\gamma$. However, two different bases produce left-equivalent codes. For two bases $\beta$ and $\gamma$, if we denote $\Pv$ the transition matrix between $\beta$ and $\gamma$, we get:
\[
\Psi_\gamma(\Cvec) = \Pv\, \Psi_\beta(\Cvec)
\]

\begin{definition}[Folding]
The application $\Fold$ turns a vector $\vv =(\vv_1\Vert ...\Vert\vv_n)\in \Fq^{mn}$ such that each $\vv_i\in\Fq^n$, into the matrix $\Fold(\vv) \eqdef \left({\vv_1}^t\Vert\dots\Vert{\vv_n^t}\right)\in\Fqmn$. The inverse map which turns a matrix into a vector is denoted by $\Unfold$.
\end{definition}

The map $\Unfold$ allows to define a matrix code $\Cmat$ thanks to a generator matrix $\Gv\in\Fq^{mk\times mn}$, whose rows are an unfolded set of matrices which forms a basis of $\Cmat$. This characterization allows to construct the associated parity check matrix $\Hv$, and define the dual of a matrix code: one can simply define $\Cmatd$ as the code generated by matrices equal to the folded rows of $\Hv$. A more formal definition follows:

\begin{definition}[Dual of a matrix code]\label{def:dual_matrix_code}
Let $\Cmat$ be a matrix code of size $m\times n$ and dimension $K$. Its dual is the matrix code of size $m\times n$ and dimension $mn-K$: $$\Cmatd = \left\lbrace \Yv\in\Fqmn \;|\; \forall\Xv\in\Cmat\; \tr (\Xv\Yv^t)=0 \right\rbrace.$$
\end{definition}

\subsection{MinRank problem}
\label{sec_MR}

The MinRank problem can be seen as the Rank Decoding problem adapted to matrix codes:

\begin{definition}[MinRank problem]
Given as input matrices $\Yv,\Mv_1,\dots,\Mv_k\in\Fq^{m\times n}$, the $\MR (q,m,n,k,r)$ problem asks to find $x_1,\dots,x_k\in\Fq$ and $\Ev\in\Fq^{m\times n}$ with $\rank\Ev\leq r$ such that $\Yv = \sum_{i=1}^k x_i\Mv_i
+ \Ev$.
\end{definition}

The decoding problem for a matrix code of size $m\times n$ and dimension $K$ is exactly the $\MR (q,m,n,K,r)$ problem. While the MinRank problem is well known, it is more uncommon to present this decoding problem from the syndrome decoding point of view.

Let be $\Xv=\sum_{i=1}^K x_i\Mv_i\in\Cmat$ with basis the set of $(\Mv_i)$ and $\Ev\in\Fqmn$ an error of small weight. The code $\Cmat$ can be represented as a vector code over $\Fq$, with generator matrix $\Gv\in\Fq^{K\times mn}$ has for rows the vectors $\Unfold(\Mv_i)$. A word to decode would have the form: $\Unfold (\Xv+\Ev)=\sum_{i=1}^K x_i\Unfold(\Mv_i)+\Unfold (\Ev)$, and the weight of the error is defined as the rank of $\Ev$.

Let $\Hv\in\Fq^{(mn-K)\times mn}$ be a parity-check matrix of the dual code. We denote by $(\hv_i)_{1\leq i\leq mn}$ its columns. Let be $\xv\in\Fq^{mn}$ a noisy codeword with error $\ev\in\Fq^{mn}$. The weight of $\ev$ is the rank of the matrix $\Fold (\ev)$.  Since one has reduced itself to a vector code, one can define the syndrome associated to an error in the same way:
$$\sv = \xv\Hv^t = \ev\Hv^t = \sum_{i=1}^{mn} e_i\hv_i^t.$$
One can deduce the associated problem: 

\begin{definition}[MinRank-Syndrome problem]
Given as input vectors $\sv,\vv_1$, $\dots,\vv_k\in\Fq^{nm-k}$, the $\MRS (q,m,n,k,r)$ problem asks to find $(e_1,\dots,e_{nm})\in\Fq^{nm}$ with $\rank\Fold (\ev)\leq r$ such that $\sv = \sum_{i=1}^k e_i\vv_i$.
\end{definition}

The two previous problems are equivalent for the very same reasons that the decoding problem and the syndrome decoding problem are in the vector codes context.

\subsection{Gabidulin codes}

Gabidulin codes were introduced by Ernst Gabidulin in 1985 \cite{G85}. These vector codes can be seen as the analog in rank metric of the Reed-Solomon codes, where the codeword is a set of evaluation points of a $q$-polynomial rather than a standard polynomial.

\begin{definition}[$q$-polynomial]
  A $q$-polynomial of $q$-degree $r$ is a polynomial in $\Fqm[X]$ of the
  form:
  $$P(X) = \sum_{i=0}^{r}p_iX^{q^i} \qquad \text{with } p_r \neq 0.$$
\end{definition}
For a $q$-polynomial $P$, we denote by $\deg_q P$ its $q$-degree.

A $q$-polynomial is also called 
\emph{linearized polynomial} since it induces an $\Fq$-linear
application due to the linearity of the Frobenius
endomorphism $x\mapsto
x^q$. 

\begin{definition}[Gabidulin code]
  Let $k,m,n\in\mathbb{N}$, such that $k\leq n\leq m$. Let
  $\gv=(g_1,\dots,g_n)\in\Fqm^n$ a vector of $\Fq$-linearly
  independent elements of $\Fqm$. The Gabidulin code
  $\mathcal{G}_{\gv} (n,k,m)$ is the vector code of parameters
  $[n,k]_{q^m}$ defined by:
$$\mathcal{G}_{\gv} (n,k,m) = \left\lbrace P(\gv)\vert\; \deg_q P < k\right\rbrace,$$ where $P(\gv)=(P(g_1),\dots,P(g_n))$ and $P$ is a $q$-polynomial.
\end{definition}

The vector $\gv$ is said to be an \emph{evaluation vector} of the Gabidulin code $\mathcal{G}_{\gv} (n,k,m)$.

Gabidulin codes are popular in cryptography because they benefit from a very efficient decoding algorithm, allowing to correct errors of rank weight up to $\left\lfloor\frac{n-k}{2}\right\rfloor$ \cite{G85}. However, their strong structure makes them difficult to hide.

\subsection{The GPT Cryptosystem}

It has been proposed in \cite{GO01} to concatenate to a generator matrix a random matrix, in order to mask a Gabidulin code more efficiently. Introduced in 1991 by Gabidulin, Paramov and Tretjakov, the GPT Cryptosystem is an adaptation of the McEliece frame to the rank metric. The first versions having been attacked by Gibson \cite{O08}. We present in Figure \ref{GPT} the last variant proposed by Gabidulin and Ourivski in \cite{GO01}. However, this scheme has been broken by the Overbeck attack, that we describe in Section \ref{str-att}.

\begin{figure}[h]
    \pcb[codesize=\scriptsize, minlineheight=0.75\baselineskip, mode=text, width=0.98\textwidth] { 
    $\KeyGen (1^\lambda)$: \\
    \pcind - Select a random Gabidulin $[n,k]_{q^m}$ code $\G$, with generator matrix $\Gv$. \\
    \pcind - Sample uniformly at random matrices: $\Sv \getsr \mathbf{GL}_k(\Fqm)$ and $\Pv \getsr \mathbf{GL}_{n+ \ell}(\Fq)$.\\
    \pcind - Sample uniformly at random: $\Xv\getsr\Fqm^{k\times \ell}$\\
    \pcind - Compute $\Gpub=\Sv (\Xv\vert\Gv)\Pv\in\Fqm^{k\times (n+\ell)}$.\\
    \pcind - Return $\pk=\Gpub$ and $\sk=(\Gv,\Sv,\Pv)$\\[\baselineskip]
  $\Encrypt (\pk,\mv)$: \\
  	\pcind \textit{Input:} $\pk = \Gpub$, $\mv\in\Fqm^k$.\\
  	\pcind - Sample uniformly at random a vector $\ev\in\Fqm^n$ of rank $r = \left\lfloor\frac{n-k}{2}\right\rfloor$.\\
  	\pcind - Return $\cv = \mv\Gpub + \ev$.\\[\baselineskip]
  	$\Decrypt (\sk,\cv)$: \\
  	\pcind - Compute $\cv\Pv^{-1}$.\\
  	\pcind - Apply the decoding algorithm of $\G$ on $\cv\Pv^{-1}$ to retrieve $\mu=\mv\Sv$.\\
    \pcind - Return $\mv=\mu\Sv^{-1}$.
  	}
\vspace{-\baselineskip}
\captionof{figure}{\footnotesize{GPT Cryptosystem}}
\label{GPT}
\end{figure}

We subsequently propose a new version of the scheme to prevent the Overbeck attack. In our scheme, we propose to turn a Gabidulin code into a matrix code before adding some random rows and columns, and multiplying it by a secret invertible matrix.


\section{Enhanced matrix codes transformation}
\label{sec:emc}

We present a general construction which defines a transformation containing a trapdoor and allowing to mask a secret matrix code.

\begin{definition}[Random Rows and Columns matrix code transformation]\label{rrcmc}
Let $m,n,K,\ell_1,\ell_2\in\mathbb{N}$. Let $\Cmat$ be a matrix code of size $m\times n$ and dimension $K$ on $\Fq$.  Let $\mathcal{B}= (\Av_1,...,\Av_K)$ be a basis of $\Cmat$. The Random Rows and Columns matrix code transformation consists in sampling uniformly at random the following matrices: $\Pv\getsr\mathbf{GL}_{m+\ell_1}(\Fq)$,
  $\Qv\getsr\mathbf{GL}_{n+\ell_2}(\Fq)$ and $K$ random matrices:
  $\Rv_i\getsr\Fqml$,
  $\Rv'_i\getsr\Fqlm$ and
  $\Rv''_i\getsr\Fqll$;  and define the matrix code whose
 basis is: $$\mathcal{RB} = \Biggl(\Pv\begin{pmatrix}
    \Av_1 & \Rv_1 \\
    \Rv'_1 & \Rv''_1
\end{pmatrix}\Qv,\dots ,\Pv\begin{pmatrix}
    \Av_{K} & \Rv_{K} \\
    \Rv'_{K} & \Rv''_{K}
\end{pmatrix}\Qv\Biggr).$$
\end{definition}

In particular, we can apply this construction to Gabidulin codes turned into matrix codes:

\begin{definition}[Enhanced Gabidulin matrix code] \label{egmc} Let
  $\G$ be a Gabidulin code $[n,k,r]$ on $\Fqm$, $\gamma$ be a $\Fq$-basis of
  $\Fqm$.  We recall that $\Psi_{\gamma}$ turns a vector $\xv\in\Fqm$
  into a matrix $\Psi_{\gamma}(\xv)$ whose columns are coordinates of
  coefficients of $\xv$ in the basis $\gamma$ (see Definition
  $\ref{g-e}$). An Enhanced Gabidulin matrix code
  $\mathcal{E}\G_{\gv}(n,k,m,\ell_1,\ell_2)$ is the matrix code
  $\Psi_{\gamma}(\G)$ on which we apply the Random Rows and Columns
  matrix code transformation presented in Definition \ref{rrcmc}.
\end{definition}

\subsubsection{Duality.}
The dual of an enhanced Gabidulin matrix code can be described as
follows.  Start from an $[n,k]$ Gabidulin code $\G$ over $\Fqm$. It is
well-known that the dual of the vector code $\G^\perp$ for the
Euclidean inner product is an $[n, n-k]$ Gabidulin code over
$\Fqm$. Moreover, given an basis $\gamma$ for $\Fqm$ and denoting by $\gamma'$
the dual basis with respect to the trace inner product in $\Fqm$, then
from \cite[Thm.~21]{R15b}, we get
\[ \Psi_{\gamma}(\G)^\perp = \Psi_{\gamma'}(\G^\perp).\] In short, the
dual of the matrix code associated to $\G$ as defined in
Definition~\ref{def:dual_matrix_code} is a matrix code associated to a
Gabidulin code.

\begin{proposition}\label{prop:dual_emc}
  Following the notation of Definitions~\ref{rrcmc} and~\ref{egmc},
  let $\Bv_1, \dots $, $\Bv_{m(n-k)}$ be an $\Fq$--basis of
  $\Psi_{\gamma}(\G)^\perp$.
  Let $\mathcal V \subseteq \Fqm^{(m + \ell_1)(n+\ell_2)}$ be defined as
  \[
    \mathcal V \eqdef \left\{
    \begin{pmatrix}
      \Rv & \mathbf{0} \\ \mathbf{0} & \mathbf{0} 
    \end{pmatrix}
    ~\bigg|~ \Rv \in \Fq^{m\times n}
    \right\}.
  \]
  Then, there exists a space $\mathcal W \subseteq
  \Fq^{(m+\ell_1) \times (n+ \ell_2)}$ such that
  \(\mathcal{E}\G_{\gv}(n,k,m,\ell_1,\ell_2)^\perp
  = {(\Pv^t)}^{-1} \Cmat {(\Qv^t)}^{-1}
  \)
  and
  \[
    \Cmat = \text{Span}_{\Fq} \left\{
      \begin{pmatrix}
        \Bv_1 & \mathbf{0} \\ \mathbf{0} & \mathbf{0}
      \end{pmatrix}
      , \dots ,
      \begin{pmatrix}
        \Bv_{m(n-k)} & \mathbf{0} \\ \mathbf{0} & \mathbf{0}
      \end{pmatrix}
    \right\} \oplus \mathcal W.
  \]
  Moreover, $\mathcal W$ is a complement subspace of $\mathcal V$ in
  $\Fq^{(m+\ell_1) \times (n+\ell_2)}$.
\end{proposition}

\begin{proof}
  Denote by $\Cmat^0$, the space
  \[
    \Cmat^0 \eqdef \text{Span}_{\Fq} \left\{
      \begin{pmatrix}
        \Bv_1 & \mathbf{0} \\ \mathbf{0} & \mathbf{0}
      \end{pmatrix}
      , \dots ,
      \begin{pmatrix}
        \Bv_{m(n-k)} & \mathbf{0} \\ \mathbf{0} & \mathbf{0}
      \end{pmatrix}
    \right\}.
  \]
  Clearly, matrices of
  ${(\Pv^t)}^{-1} \Cmat^0 {(\Qv^t)}^{-1}$ lie in
  \(\mathcal{E}\G_{\gv}(n,k,m,\ell_1,\ell_2)^\perp\).
  Moreover, any matrix of the shape ${(\Pv^t)}^{-1}
  \begin{pmatrix}
    \Bv & \mathbf 0 \\ \mathbf 0 & \mathbf 0 
  \end{pmatrix}
  {(\Qv^t)}^{-1} $ lying in
  \(\mathcal{E}\G_{\gv}(n,k,m,\ell_1,\ell_2)^\perp\) should satisfy
  $\Bv \in \Psi_{\gamma}(\G)^\perp$.  
  Consequently, $\mathcal V \cap \mathcal W = 0$ and, for dimensional
  reasons, they should be complement subspaces.
\end{proof}

More generally, one can prove that any matrix code with the above
shape is the dual of an enhanced Gabidulin code. Namely, starting with
a matrix Gabidulin code ``extended by zero'' to which we add a random
complement subspace of the aforementioned space $\mathcal V$ and which
we left and right multiply by invertible matrices, we get the dual of
an enhanced Gabidulin code.

\subsubsection{Indistinguishability of Enhanced Gabidulin matrix codes.}
The security of the
scheme we introduce later is based on the difficulty of distinguishing
a random matrix code from a code as defined above.

We formally define in this section the problem of distinguishing an
Enhanced Gabidulin matrix code from a random matrix code, on which the security of the EGMC-McEliece encryption scheme that we
present below is based, and the associated search problem. We conjecture that
both problems are difficult to solve.

\begin{definition}[$\EGMC (k,m,n,\ell_1,\ell_2)$ distribution]\label{def:egmc_distribution}
  Let $k,m,n,\ell_1,\ell_2\in\mathbb{N}$ be such that $k\leq n\leq
  m$. The Enhanced Gabidulin Matrix Code distribution
  $\EGMC (k,m,n,\ell_1,\ell_2)$ samples a vector $\gv\getsr\Fqm^n$, a
  basis $\gamma\getsr\mathcal{B}(\Fqm)$,
  $\Pv\getsr\mathbf{GL}_{m+\ell_1}(\Fq)$,
  $\Qv\getsr\mathbf{GL}_{n+\ell_2}(\Fq)$ and $km$ random matrices:
  $\Rv_i\getsr\Fqml$,
  $\Rv'_i\getsr\Fqlm$ and
  $\Rv''_i\getsr\Fqll$;
  computes $\mathcal{B}= (\Av_1,...,\Av_{km})$ a basis of the matrix
  code $\Psi_{\gamma}(\mathcal{G})$, and outputs the matrix code with basis: $$\mathcal{RB} = \Biggl(\Pv\begin{pmatrix}
    \Av_1 & \Rv_1 \\
    \Rv'_1 & \Rv''_1
\end{pmatrix}\Qv,...,\Pv\begin{pmatrix}
    \Av_{km} & \Rv_{km} \\
    \Rv'_{km} & \Rv''_{km}
\end{pmatrix}\Qv\Biggr).$$
\end{definition}

\begin{definition}[$\EGMC $-Indistinguishability problem]
\label{egmc-ind}
Given a matrix code $\Cmat$ of size $(m+\ell_1)\times (n+\ell_2)$ and dimension $mk$, the Decisional $\EGMC $-Indistinguishability $(k,m,n,\ell_1,\ell_2)$ problem asks to decide with non-negligible advantage whether $\Cmat$ sampled from the $\EGMC (k,m,n,\ell_1,\ell_2)$ distribution or the uniform distribution over the set of $\Fq$-subspaces of $\Fq^{(m+\ell_1)\times (n+\ell_2)}$ of dimension $mk$.
\end{definition}

\begin{definition}[$\EGMC $-Search problem]
\label{egmc-s}
Let $k,m,n,\ell\in\mathbb{N}$ be such that $k\leq n\leq m$, and $\Cmat$ sampled from the $\EGMC (k,m,n,\ell)$ distribution. The $\EGMC $-Search problem asks to retrieve the basis $\gamma\in\mathcal{B}(\Fqm)$ and the evaluation vector $\gv\in\Fqm^n$ used to construct $\Cmat$.
\end{definition}


\textbf{Claim.}  {\em Given a vector code $\Cvec \subseteq \Fqm^n$ and an $\Fq$--basis $\gamma$ of $\Fqm$. Then, 
   the matrix code
  $\Cmat = \Psi_{\gamma}(\Cvec)$ is
  distinguishable from a random matrix code in polynomial time.}


Indeed, even without knowing $\gamma$, the $\Fqm$--linear structure
can be detected by computing the left stabilizer algebra of
$\Cmat$. That is to say
  \[
    \text{Stab}_L(\Cmat) \eqdef \left\{\Pv \in \Fq^{m \times m} ~|~ \forall \Cv \in \Cmat,\ \Pv \Cv \in \Cmat \right\}.
  \]
  For a random matrix code, with high probability this algebra has dimension $1$ 
  and only contains the matrices $\lambda \Iv$ where $\lambda \in \Fq$
  and $\Iv$ denotes the $m\times m$ identity matrix. For a code $\Cmat$ which comes from an $\Fqm$--linear code, the
  algebra $\text{Stab}_L (\Cmat)$ contains a sub-algebra isomorphic to
  $\Fqm$ and hence has dimension at least $m$. Since the computation
  of the stabilizer algebra can be done by solving a linear system,
  this yields a polynomial time distinguisher.
  
  \begin{remark}
    Actually, using tools described in \cite{CDG20} it is even
    possible to do more than distinguishing and recover a description
    of $\Cvec$ as a vector code.
  \end{remark}
  
\subsubsection{Alternative approach: deleting rows and columns.}
We could have considered another masking procedure which consists in deleting rows and columns of the matrix code rather than adding to them random rows and columns. Deleting rows or columns is also an option which blurs the (right and left) stabilizer algebras. If some columns are removed, we obtain a punctured Gabidulin code, which remains a Gabidulin code, and hence can be decoded. If some rows are also removed, then the decryption will require to correct both errors and erasures, which forces to reduce the amount of errors during the encryption process and then reduces the security with respect to generic decoding attacks.


\section{Generic McEliece and Niederreiter frames based on MinRank}
\label{new_frames}

The McEliece frame is a generic process to construct code-based encryption schemes, proposed in \cite{M78}. It consists in masking a generator matrix of a vector code for which we know an efficient decoding algorithm with a secret operation. For an opponent who does not know the secret transformation, decrypting the ciphertext is as difficult as decoding a general linear code.

We propose here an adaptation of the McEliece frame to matrix codes: rather than hiding the generator matrix of a vector code, we propose to hide a basis of a secret matrix code, for which there exists an efficient decoding algorithm. The resulting scheme can be found in Figure \ref{mr-mce}.

\begin{figure}[h]
    \pcb[codesize=\scriptsize, minlineheight=0.75\baselineskip, mode=text, width=0.98\textwidth] { 
    $\KeyGen (1^\lambda)$: \\
    \pcind - Select a matrix code $\Cmat$ of size $m\times n$ and dimension $K$ on $\Fq$, with an efficient algorithm capable of decoding up to $r$ errors.\\
    \pcind - Let $\mathcal{T}$ a transformation which turns a matrix code into an other one with a trapdoor. Define the code $\Cmat' = \mathcal{T}(\Cmat)$.\\
    \pcind - Compute $\mathcal{B}=(\Mv_1,...,\Mv_K)$ a basis of $\Cmat'$.\\
    \pcind - Return $\pk=\mathcal{B}$ and $\sk=\left(\Cmat, \mathcal{T}^{-1}\right)$.\\[\baselineskip]
  $\Encrypt (\pk,\mu)$: \\
  	\pcind \textit{Input:} $\pk = (\Mv_1,...,\Mv_K)$, $\mu\in\Fq^K$.\\
  	\pcind - Sample uniformly at random a matrix $\Ev\in\Fq^{m\times n}$ such that $\Rank\Ev\leq r$.\\
  	\pcind - Return $\Yv = \sum_{i=1}^K \mu_i\Mv_i + \Ev$.\\[\baselineskip]
  	$\Decrypt (\sk,\Yv)$: \\
  	\pcind - Compute $\tilde{\Yv}=\mathcal{T}^{-1}(\Yv)$.\\
  	\pcind - Apply the decoding algorithm of $\Cmat$ on the matrix $\tilde{\Yv}$ to retrieve the message $\mu$.
  	}
\vspace{-\baselineskip}
\captionof{figure}{\footnotesize{MinRank-McEliece encryption frame}}
\label{mr-mce}
\end{figure}

\textbf{Comments.} Not just any transformation can be chosen: it must be compatible with the decoding of the noisy matrix $\tilde{\Yv}$.\\
An opponent which does not know the original code $\Cmat$ must solve a general instance of the MinRank problem to retrieve the message $\mu$, that is at least as difficult as the decoding problem in rank metric. Therefore, it guarantees a security at least as high as that of the standard McEliece frame for the same sizes of ciphertext and public key.

The Niederreiter frame is a variant of the McEliece frame for code-based encryption, which consists in hiding a parity check matrix rather than a generator matrix. It allows to obtain syndromes as ciphertexts, which are shorter than noisy codewords for an equivalent security, since the MinRank problem and the MinRank-Syndrome problem are equivalent. 
The resulting scheme can be found in Figure \ref{mr-nie}.

\begin{figure}[h]
    \pcb[codesize=\scriptsize, minlineheight=0.75\baselineskip, mode=text, width=0.98\textwidth] { 
    $\KeyGen (1^\lambda)$: \\
     \pcind - Select a matrix code $\Cmat$ of size $m\times n$ and dimension $K$ on $\Fq$, with an efficient algorithm capable of decoding up to $r$ errors.\\
    \pcind - Let $\mathcal{T}$ a transformation which turns a matrix code into an other one with a trapdoor. Define the code $\Cmat' = \mathcal{T}(\Cmat)$.\\
    \pcind - Compute $\bar{\Hv}\in\mathcal{M}_{(mn-K)\times mn}(\Fq)$ a parity check matrix of $\Cmat'$.\\
    \pcind - Return $	\mathbf{pk} = \bar{\Hv}$ and $\sk=\left(\Cmat, \mathcal{T}^{-1}\right)$.\\[\baselineskip]
  $\Encrypt (\pk,\mu)$: \\
  \pcind \textit{Input:} $\pk=\bar{\Hv}$, a message $\mu \in \Fq^{nm}$ such that $\rank \Fold(\mu)\leq r$.\\
  	\pcind - For every integer $i$ from 1 to $nm$, let $\hv_i$ the i-th column of $\bar{\Hv}$. \\
	\pcind - Return $\cv = \sum_{i = 1}^{nm} \mu_i \hv_i^t$.\\[\baselineskip]
	$\Decrypt (\sk,\cv)$:\\
    \pcind \textit{Input:} $\sk=(\Cmat, \Pv, \Qv)$, $\cv\in\Fqm^{nm-K}$.\\
    \pcind - Find any $\yv\in\Fq^{nm}$ such that $\cv =\sum_{i=1}^{nm}\bar{y}_i \hv_i^t$.\\
    \pcind - Let $\Yv = \mathcal{T}^{-1}(\Fold (\yv))$. Apply the decoding algorithm of $\Cmat$ on the matrix $\Yv$ to retrieve the error $\mu$.
  	}
\vspace{-\baselineskip}
\captionof{figure}{\footnotesize{MinRank-Niederreiter encryption frame}}
\label{mr-nie}
\end{figure}

\textbf{Comments.} Since only the original matrix code is equipped with an efficient decoding algorithm, it is necessary to retrieve a noisy word whose syndrome is the ciphertext $\cv$. Unlike the previous scheme, the message $\mu$ corresponds to the error of the noisy codeword.



\section{New encryption schemes}
\label{nes}

\subsection{EGMC-McEliece encryption scheme}

We apply the encryption frames we defined above in Section \ref{new_frames} to matrix Gabidulin codes, by using the trapdoor presented in Definition \ref{rrcmc}. We obtain an encryption scheme whose public key is an Enhanced Gabidulin Matrix code. The secret key contains the original Gabidulin code $\G$, the basis $\gamma$ on which the secret code $\G$ has been extended and the trapdoor. The resulting scheme is presented in Figure~\ref{sch-mce}.

\begin{figure}[h]
    \pcb[codesize=\scriptsize, minlineheight=0.75\baselineskip, mode=text, width=0.98\textwidth] { 
    $\KeyGen (1^\lambda)$: \\
    \pcind - Select a random $[m,k]_{q^m}$ Gabidulin code $\G$, with an efficient algorithm capable of decoding up to $\left\lfloor\frac{m-k}{2}\right\rfloor$ errors.\\
    \pcind - Sample uniformly at random a basis $\gamma\getsr\mathcal{B}(\Fqm)$.\\
    \pcind - Compute a basis $(\Av_1,...,\Av_{km})$ of the code $\Psi_\gamma (\G)$.\\
    \pcind - For $i$ in range 1 to $km$, sample uniformly at random: $\Rv_i\getsr\Fqml$, $\Rv'_i\getsr\Fqlm$ and $\Rv''_i\getsr\Fqll$.\\
    \pcind - Define the matrix code $\Cmat$ as explained in Definition \ref{def:egmc_distribution}.\\
    \pcind - Sample uniformly at random matrices $\Pv \getsr \mathbf{GL}_{m+\ell_1}(\Fq)$ and $\Qv \getsr \mathbf{GL}_{m+\ell_2}(\Fq)$.\\
    \pcind - Define the code $\Cmat' = \Pv\Cmat\Qv$.\\
    \pcind - Let $\mathcal{B}=(\Mv_1,...,\Mv_{km})$ be a basis of $\Cmat'$.\\
    \pcind - Return $\pk=\mathcal{B}$ and $\sk=(\G, \gamma, \Pv, \Qv)$.\\[\baselineskip]
  $\Encrypt (\pk,\mu)$: \\
  	\pcind \textit{Input:} $\pk = (\Mv_1,...,\Mv_{km})$, $\mu\in\Fq^{km}$.\\
  	\pcind - Sample uniformly at random a matrix $\Ev\in\Fq^{(m+\ell_1)\times (m+\ell_2)}$ such that $\Rank\Ev\leq r$.\\
  	\pcind - Return $\Yv = \sum_{i=1}^{km} \mu_i\Mv_i + \Ev$.\\[\baselineskip]
  	$\Decrypt (\sk,\Yv)$: \\
   \pcind - Compute $\Pv^{-1}\Yv\Qv^{-1}$. Truncate the $\ell_1$ last rows to obtain $\Mv\in\Fq^{m\times (m+\ell_2)}$.\\
  	\pcind - Compute $\Psi_\gamma^{-1}(\Mv)\in\Fq^{m+\ell_2}$.  Let $\yv$ be the first $m$ entries.\\
  	\pcind - Apply the decoding algorithm of $\G$ on the word $\yv$, it returns an error vector $\ev\in\Fqm^m$.\\
  	\pcind - Compute $\Ev' = \Psi_\gamma(\ev)$.\\
  	\pcind - Solve the linear system: $$ \Yv - (\Ev'|\Nv)\Qv = \sum_{i = 1}^{km} \mu_i \Mv_i $$ whose $(k+\ell_2)m$ unknowns in $\Fq$ are the $\mu_i$ and the remaining part of the error $\Nv$ of size $m\times\ell_2$.\\
  	\pcind - Return $\hat\mu$ the solution of the above system.
  	}
\vspace{-\baselineskip}
\captionof{figure}{\footnotesize{EGMC-McEliece encryption scheme}}
\label{sch-mce}
\end{figure}

\textbf{Comments.} The frame is valid since the multiplication by an invertible matrix makes the rank invariant. We can easily see that $\tilde{\Yv}$ is a noisy codeword of $\Cmat$, whose error is equal to $\Pv^{-1}\Ev\Qv^{-1}$. Since the multiplication by $\Pv^{-1}$ and $\Qv^{-1}$ does not change the rank of the matrix, we still have: $\Rank\Pv^{-1}\Ev\Qv^{-1}\leq r$.\\
The Enhanced code $\Cmat$ is not decodable on all its coordinates: the decoding algorithms allow to decode only the $n$ first coordinates, but the random noise map prevents from decoding the others. However, these $\ell$ coordinates can be recovered by solving a linear system of equations. On the other hand, using a matrix code, less structured than a vector code, increases the complexity of the attacks.\\
An opponent who wants to decrypt a ciphertext without retrieving the secret key has to solve a MinRank instance. Consequently, we choose in practice the same value for $\ell_1$ and $\ell_2$ in order to obtain square matrices, the algebraic attacks against the MinRank problem being less efficient in this case.

\begin{proposition}[Decryption correctness]
If the weight $r$ of the error is at most the decoding capacity $\left\lfloor\frac{m-k}{2}\right\rfloor$ of the code $\G$, the decryption algorithm outputs the correct message $\mu$.
\end{proposition}

\begin{proof}
  Let 
  $\Phi : \Fq^{(m+\ell_1) \times (m+\ell_2)} \rightarrow \Fq^{m \times
    m}$ be the map truncating matrices by removing the $\ell_1$ last
  rows and $\ell_2$ last columns. Then,
  \begin{equation}\label{eq:decryption}
    \ \tilde{\Yv} = \Phi(\Pv^{-1}\Yv\Qv^{-1}) = \sum_i \mu_i
    \Phi(\Pv^{-1}\Mv_i\Qv^{-1}) + \Phi(\Pv^{-1}\Ev\Qv^{-1}).
  \end{equation} Since
  $\Phi(\Pv^{-1}\Mv_1\Qv^{-1}), \dots, \Phi(\Pv^{-1}\Mv_{km}\Qv^{-1})$
  form an $\Fq$--basis of $\Cmat$ and\\
  $\Rank (\Phi(\Pv^{-1}\Ev\Qv^{-1}))\leq r$, applying decoding
  algorithm to \eqref{eq:decryption} yields $\mu$.  \qed
\end{proof}

\begin{proposition}\label{ow-cpa-mce}
    Under the assumption that there exists no PPT algorithm to solve the MinRank problem with non negligible probability and no PPT distinguisher for the $\EGMC $-Indistinguishability problem with non negligible advantage, then the scheme presented Figure \ref{sch-mce} is OW-CPA.
\end{proposition}

\begin{proof}
    Suppose there exists an efficient PPT decoder $\mA_{EGMC}$ of the EGMC-McEliece encryption scheme which takes as input the ciphertext and the public key, with a non negligible probability $\varepsilon (\lambda)$. Using this algorithm, we construct the following distinguisher $\mD$ for the $\EGMC $-Indistinguishability problem:

\begin{algorithmic}
\State \textbf{Input}: $\mathcal{B}= (\Mv_1,...,\Mv_{km})$
\State $\mu\getsr\Fq^{km}$
\State $\Ev\getsr\{\Xv\in\Fq^{(m+\ell_1)\times (m+\ell_2)}\vert\text{ rank }\Xv\leq r\}$
\State $\Tilde{\mu}\leftarrow\mA_{EGMC} (\mathcal{B},\sum\mu_i\Mv_i +\Ev)$
\If{$\mu=\Tilde{\mu}$}
\State \Return $1$
\Else
\State \Return $0$
\EndIf
\end{algorithmic}

Let $\mA_{MR}$ be an adversary against the MinRank problem. In the case where $\mathcal{B}$ is not a basis of an Enhanced Gabidulin matrix code, retrieving $\Ev$ is equivalent to solving a random MinRank instance. Then: \begin{align*}
    \Adv(\mD) &= \big\vert \Pr (\mD(\mathcal{B})=1\vert\mathcal{B}\text{ uniformly random})- \Pr (\mD(\mathcal{B})=1\vert\mathcal{B}\text{ basis of EGMC})\big\vert \\
    &= \big\vert \Succ (\mA_{MR}) - \Succ (\mA_{EGMC})\big\vert,
\end{align*}
since decrypting a ciphertext is strictly equivalent to retrieving the error $\Ev$ in the case where $\mathcal{B}$ is the basis of an EGMC. We deduce that: $$\Succ (\mA_{EGMC})\leq \Succ (\mA_{MR}) + \Adv(\mD).$$
Since $\Succ (\mA_{EGMC})\ge\varepsilon (\lambda)$, then $\Succ (\mA_{MR})\ge\frac{\varepsilon (\lambda)}{2}$ or $\Adv(\mD)\ge\frac{\varepsilon (\lambda)}{2}$, that are non negligible probabilities. We have proven that there exists either an efficient PPT algorithm to solve the MinRank problem, or an efficient PPT distinguisher to solve the EGMC-Indistinguishability problem, which allows to conclude.\qed
\end{proof}

\subsection{EGMC-Niederreiter encryption scheme}

It is also possible to adapt the above ideas to the Niederreiter frame, the dual version of the McEliece frame. It is then necessary to adapt the Syndrome Decoding approach to matrix codes, in the same way as what is presented in Subsection~\ref{sec_MR}. The resulting scheme can be found in Figure \ref{sch-nied}.

\begin{figure}[h!]
    \pcb[codesize=\scriptsize, minlineheight=0.75\baselineskip, mode=text, width=0.98\textwidth] { 
    $\KeyGen (1^\lambda)$: \\
    \pcind - Select a random $[m,k]_{q^m}$ Gabidulin code $\G$, with an efficient algorithm capable of decoding up to $\left\lfloor\frac{m-k}{2}\right\rfloor$ errors.\\
    \pcind - Sample uniformly at random a basis $\gamma\getsr\mathcal{B}(\Fqm)$.\\
    \pcind - Compute a basis $(\Av_1,...,\Av_{km})$ of the code $\Psi_\gamma (\G)$.\\
    \pcind - For $i$ in range 1 to $km$, sample uniformly at random: $\Rv_i\getsr\Fqml$, $\Rv'_i\getsr\Fqlm$ and $\Rv''_i\getsr\Fqll$.\\
    \pcind - Define the matrix code $\Cmat$ as explained above.\\
    \pcind -Sample uniformly at random matrices $\Pv \getsr \mathbf{GL}_{m+\ell_1}(\Fq)$ and $\Qv \getsr \mathbf{GL}_{m+\ell_2}(\Fq)$.\\
    \pcind - Define the code $\Cmat' = \Pv\Cmat\Qv$.\\
    \pcind - Compute $\bar{\Hv}\in\Fq^{((m+\ell_1)(m+\ell_2)-mk)\times (m+\ell_1)(m+\ell_2)}$ a parity check matrix of $\Cmat'$.\\
	\pcind - Return $	\mathbf{pk} = \bar{\Hv} \,;\, \mathbf{sk} = (\gamma, \G, \Qv)$\\[\baselineskip]
  $\Encrypt (\pk,\mu)$: \\
  	\pcind \textit{Input:} $\pk=\bar{\Hv}$, a message $\mu \in \Fq^{(m+\ell_1)(m+\ell_2)}$ such that $\rank \Fold(\mu)\leq r$.\\
  	\pcind - For every integer $i$ from 1 to $(m+\ell_1)(m+\ell_2)$, let $\hv_i$ be the i-th column of $\bar{\Hv}$. \\
	\pcind - Return $\cv = \sum_{i = 1}^{(m+\ell_1)(m+\ell_2)} \mu_i \hv_i^t$.\\[\baselineskip]
	$\Decrypt (\sk,\cv)$:\\
	\pcind \textit{Input:} $\mathbf{sk} = (\gamma, \G, \Qv)$, $\cv\in\Fq^{(m+\ell_1)(m+\ell_2)-mk}$.\\
	\pcind - Find any $\bar{\yv}\in\Fq^{(m+\ell_1)(m+\ell_2)}$ such that $\cv =\sum_{i=1}^{(m+\ell_1)(m+\ell_2)}\bar{y}_i \hv_i^t$.\\
    \pcind - Let $\Yv=\Fold(\bar{y})$. Compute $\Pv^{-1}\Yv\Qv^{-1}$. Truncate the $\ell_1$ last rows to obtain $\Mv\in\mathcal{M}_{m\times (m+\ell)}(\mathbb{F}_q)$.\\
    \pcind - Let $\yv$ be the first $m$ coordinates of $\Psi_\gamma^{-1}(\Mv)\in\Fq^{m+\ell_2}$.\\
	\pcind - Let $\yv=\mv\Gv+\ev$, with $\rank (\ev) \leq r$. Apply the decoding algorithm of $\G$ on the word composed of the $m$ first coordinates.\\
	\pcind - Let $\Ev= \Psi_\gamma(\ev)$. Then $\Ev=(\Ev'|\Nv)\in\mathcal{M}_{m\times (m+\ell_2)}(\mathbb{F}_q)$ is a matrix whose the $nm$ coefficients of $\Ev'$ are known, and $\Nv$ is the remaining part of the error.\\
	\pcind - Solve the linear system $\cv = \sum_{i=1}^{(m+\ell_1)(m+\ell_2)}\bar{e}_i \hv_i^t$ to find the $\ell m$ values of $\Nv$, with $\bar{\ev}=\Unfold (\Ev)$ with unknown some coefficients.\\
	\pcind - Return $\hat\mu=\Unfold(\Ev)$.
}
\vspace{-\baselineskip}
\captionof{figure}{\footnotesize{EGMC-Niederreiter encryption scheme}}
\label{sch-nied}
\end{figure}

\begin{proposition}[Decryption correctness]
If the weight $r$ of the error is under the decoding capacity $\left\lfloor\frac{m-k}{2}\right\rfloor$ of the code $\G$, the decryption algorithm outputs the correct message $\mu$.
\end{proposition}

\begin{proof}
Let $\bar{\Gv}\in\mathcal{M}_{mk\times (m+\ell_1)(m+\ell_2)}(\Fq)$ be a generator matrix of the associated unfolded $\Cmat'$ vector code, and $\bar{\Hv}\in\mathcal{M}_{((m+\ell_1)(m+\ell_2)-k)\times (m+\ell_1)(m+\ell_2)}(\Fq)$ a parity check matrix. Let $\bar{\yv}\in\Fq^{(m+\ell_1)(m+\ell_2)}$ be such that $\cv =\sum_{i=1}^{(m+\ell_1)(m+\ell_2)}\bar{y}_i\Mv_i$. Then:
$$\Unfold (\cv) = \sum_{i=1}^{(m+\ell_1)(m+\ell_2)}\bar{y}_i\Unfold(\Mv_i)=\bar{\yv}\bar{\Hv}^t = \bar{\ev}\bar{\Hv}^t$$
where $\bar{\ev}=\Unfold ({\Psi_\gamma (\ev)})$.

The $nm$ first values of $\bar{\ev}$ can be retrieved by decoding a word of the noisy Gabidulin code $\Psi_\gamma^{-1}(\Cmat')$. The rank $r$ of the error must not exceed the decoding capacity.

The $\ell_2 m$ last values can be retrieved by solving a linear system of $(m-k+\ell_2)m$ equations, which is possible since $k\leq m$. Multiplying by $\Qv^{-1}$ allows to retrieve the original message $\mu$.\qed
\end{proof}

\begin{proposition}
    Under the assumption that there exists no PPT algorithm to solve the MinRank-Syndrome problem with non negligible probability and no PPT distinguisher for the $\EGMC $-Indistinguishability problem with non negligible advantage, then the scheme presented Figure \ref{sch-nied} is OW-CPA.
\end{proposition}

\begin{proof}
    This proof is similar to that of Proposition \ref{ow-cpa-mce}. Let $\mA_{EGMC-N}$ be an efficient PPT adversary for the EGMC-Niederreiter encryption scheme. By considering the following distinguisher for the $\EGMC$-Indistinguishability problem:

\begin{algorithmic}
\State \textbf{Input}: $\bar{\Hv}\in\Fq^{((m+\ell_1)(m+\ell_2)-mk)\times (m+\ell_1)(m+\ell_2)}$
\State $\mu\getsr\{\mv\in\Fq^{(m+\ell_1)(m+\ell_2)}\vert\rank \Fold(\mv)\leq r\}$
\State $\Tilde{\mu}\leftarrow\mA_{EGMC-N} (\bar{\Hv},\sum\mu_i \hv_i^t)$
\If{$\mu=\Tilde{\mu}$}
\State \Return $1$
\Else
\State \Return $0$
\EndIf
\end{algorithmic}

one demonstrates by identical reasoning that there exists either an efficient PPT algorithm to solve the MinRank-Syndrome problem, or an efficient PPT distinguisher to solve the EGMC-Indistinguishability problem.\qed
\end{proof}

\subsection{Algorithms complexity}

The most costly step for the key generation is the multiplication by the matrices $\Pv$ and $\Qv$. It involves computing $km$ products of the form $\Pv \Av_i \Qv$, which requires $km ((m + \ell_1)^2 (m + \ell_2) + (m + \ell_2)^2 (m + \ell_1))$ multiplications in $\Fq$. For the Neiderreiter variant, we must also consider the cost of computing the matrix $\bar{\Hv}$, which can be done using Gaussian elimination, giving a complexity of $\mathcal{O}(((m + \ell_1)(m + \ell_2))^3)$ operations in $\Fq$.

The encryption step is the cheapest as we only need to compute the ciphertext $\cv$ by adding $(m + \ell_1)(m + \ell_2)$ elements of $\Fq^{(m + \ell_1)(m + \ell_2) - km}$, each requiring $(m + \ell_1)(m + \ell_2) - km$ multiplications in $\Fq$ for the multiplication by $\mu_i$ in the Neiderreiter variant. For the McEliece variant, computing the $\mu_i\Mv_i$ requires $km(m + \ell_1)(m + \ell_2)$ multiplications in $\Fq$.

The decryption process requires solving one (or two in the case of the Neiderreiter variant) linear system with $(m + \ell_1)(m + \ell_2)$ unknowns in $\Fq$. This gives a complexity of $\mathcal{O}(((m + \ell_1)(m + \ell_2))^3)$ operations in $\Fq$.

\section{Security analysis}\label{sec:secu}

This section deals with the security of the schemes proposed in Section \ref{nes}. There exist two types of attacks. The first one are the structural attacks which consist in retrieving the structure of the secret code from the information given in the public key. It amounts to recovering the secret key by solving the $\EGMC$-Search problem (see Definition \ref{egmc-s}), and efficiently decrypting any ciphertext thanks to the underlying decoding algorithm of the secret code. The second one includes attacks which try to recover the message from the ciphertext and the public key, that is equivalent to solving an instance of the MinRank problem.

\subsection{Attacks on the key}

We recall that an opponent who wants to retrieve the secret key has to solve the $\EGMC$-Search problem:\\
\textbf{Instance: }A matrix code $\Cmat$ sampled from the $\EGMC (k,m,n,\ell_1, \ell_2)$ distribution.\\
\textbf{Problem: }Retrieve the basis $\gamma\in\mathcal{B}(\Fqm)$ and the evaluation vector $\gv\in\Fqm^n$ of the Gabidulin code $\G$ used to construct $\Cmat$.

In this section, we present two algorithms for solving the $\EGMC$ decision problem. Just like other distinguishers for structured codes in rank metric (e.g. LRPC codes), there are two types of attacks. The first one is of combinatorial nature and tries to detect the $\Fqm$-linear structure, whereas the second one is an algebraic distinguisher that is inspired from the Overbeck attack.

Before that, we present polynomial distinguishers for vector and matrix plain Gabidulin codes, that are not directly usable for $\EGMC$ codes considered in our schemes, since they are enhanced with perturbing random rows and columns that increase the cost of a distinguisher.

\subsubsection{Distinguisher for vector Gabidulin codes.}
\label{str-att}

Due to their strong algebraic structure, Gabidulin codes can be easily distinguished from random linear codes. As it has been proven, the security of the schemes relies on the difficulty to distinguish an Enhanced Gabidulin vector code (see Proposition \ref{ow-cpa-mce}).

Let $\xv = (x_1,\dots,x_n)\in\Fqm^n$. For any $i\in\{0,\dots,m-1\}$, we define: $$\xv^{[i]}=(x_1^{q^i},\dots,x_n^{q^i}).$$
This definition naturally extends to codes. For a vector space $E\subset\Fqm^n$, we write $E^{[1]}$ the image of $E$ by the Frobenius application. We can generalize this notation to $E^{[i]}$ for $i$ compositions of the Frobenius.

\begin{definition}
    Let $\C$ be an $[n,k]_{q^m}$ linear code. We define the $f$-th Frobenius sum of $\C$ as: $$\Lambda_f (\C) = \C + \C^{[1]} + \cdots + \C^{[f]}.$$
\end{definition}

If $\Gv$ is a generator matrix of $\C$, then a generator matrix for $\Lambda_f (\C)$ is: $$\begin{pmatrix}
    \Gv \\
    \Gv ^{[1]}\\
    \vdots \\
    \Gv ^{[f]}
\end{pmatrix} \in \Fqm^{(f+1)k\times n}.$$
For convenience, we abusively denote this matrix as $\Lambda_f (\Gv).$

\begin{proposition}[\cite{GOT18}]
    Let $\G$ be an $[n,k]_{q^m}$ Gabidulin code. For any $f\ge 0$: $$\dim\Lambda_f (\G) = \min \{n,k+f\}.$$
    Let $\C$ be a random $[n,k]_{q^m}$ linear code. For any $f\ge 0$, we have with high probability: $$\dim\Lambda_f (\C) = \min \{n,k(f+1)\}.$$
\end{proposition}

\textit{Remark 1.} Let $f\in\{0,\dots,n-k\}$. If $\G$ is a Gabidulin code $[n,k]_{q^m}$ with evaluation vector $\gv$, then $\Lambda_f (\G)$ is a Gabidulin code $[n,n-1]_{q^m}$ with the same evaluation vector $\gv$.

\textit{Remark 2.} Another way to distinguish a Gabidulin code from a random one is to observe the intersection: $\C\cap\C^{[1]}=\{0\}$ with high probability for a random code, whereas $\dim (\G\cap\G^{[1]})=k-1$ for any Gabidulin code.

Since the rank of a matrix can be computed in polynomial time, this noteworthy behavior allows to easily distinguish a Gabidulin code from a random code. It has been exploited by Overbeck. See \cite{O05,O08} for some structural attacks against cryptographic schemes based on Gabidulin codes.



\subsubsection{Distinguisher for matrix Gabidulin codes.}

As already mentioned in the
preliminary part, starting from an $\Fqm$--linear code
$\Cvec \subseteq \Fqm^n$, the map $\Psi_\gamma$ sending it on a matrix
code is not enough to hide it from a random matrix code since
it has a non trivial left stabilizer algebra.

When $n = m$, the following statement suggests the
existence of another detectable structure.

\begin{lemma}
  Suppose that $n= m$ and let $\gv \in \Fqm^m$ whose entries are
  $\Fq$--linearly independent. let $\gamma$ be an $\Fq$--basis of
  $\Fqm$. Then, $\Psi_{\gamma} (\Gabcode{\gv}{m,k,m})$ has a right stabilizer
  algebra of dimension $\geq m$.
\end{lemma}

\begin{proof}
  A $q$--polynomial $P$ of $q$--degree $< k$ induces an
  $\Fq$--endomorphism of $\Fqm$ and
  $\Psi_\gamma((P(g_1), \dots, P(g_m)))$ is the matrix representation
  of this endomorphism from the basis $\gv = (g_1, \dots, g_m)$ to the
  basis $\gamma$.
  Then, observe that the space of $q$--polynomials of degree $< k$ is
  stable by right composition by any $q$--polynomial $\alpha X$ for
  $\alpha \in \Fqm$.  Indeed, for any
  $P = p_0X + p_1 X^q + \cdots + p_{k-1}X^{q^{k-1}}$ with $q$--degree
  $<k$, we have
  \[
    P \circ \alpha X = p_0 \alpha X + p_1 \alpha^q X^q + \cdots +
    p_{k-1} \alpha^{q^{k-1}}X^{q^{k-1}}
  \]
  and the resulting $q$--polynomial has the same degree.

  The stability of this space of $q$--polynomials by right composition
  by $\alpha X$ entails that $\Psi_{\gamma}(\Gabcode{\gv}{m,k,m})$ is
  stabilized on the right by the matrix representing the
  multiplication by $\alpha$ (regarded as an $\Fq$--endomorphism of
  $\Fqm$) in the basis $\gv$.
  
  In summary, the right stabilizer algebra contains a sub-algebra
  isomorphic to $\Fqm$ and hence has dimension larger than or equal to
  $m$. \qed
\end{proof}

The goal of the masking we propose is to mask both left and right
$\Fqm$--linear structure of a Gabidulin code, leading to trivial
left and right stabilizer algebras.

\subsubsection{A combinatorial distinguisher detecting the
  $\Fqm$--linear structure.}
Suppose from now on that we are given a code $\Cmat$ which is an
enhanced Gabidulin matrix code. We keep the notation of
Definitions~\ref{egmc} and \ref{def:egmc_distribution}.  That is to
say our code that we denote $\Cpub$ is described by a basis:
\begin{equation}\label{eq:RBpub}
  \RBpub = \Biggl(\Pv\begin{pmatrix}
    \Av_1 & \Rv_1 \\
    \Rv'_1 & \Rv''_1
\end{pmatrix}\Qv,\dots,\Pv\begin{pmatrix}
    \Av_{km} & \Rv_{km} \\
    \Rv'_{km} & \Rv''_{km}
\end{pmatrix}\Qv\Biggr),
\end{equation}
where the $\Av_{ij}$'s are an $\Fq$--basis of
$\Psi_{\gamma}(\Gabcode{\gv}{m,k,n})$, $\Pv, \Qv$ are random
nonsingular matrices and the $\Rv_i$, $\Rv'_i$, $\Rv_i''$s are random
matrices with respective sizes $m \times \ell_2$, $\ell_1 \times n$
and $\ell_1 \times \ell_2$. Roughly speaking, this is a Gabidulin
matrix code, enhanced with $\ell_1$ random rows and $\ell_2$ random
columns. The matrix $\Pv$ (resp. $\Qv$) ``mixes'' rows (resp. columns)
from the Gabidulin code with random ones.

We also introduce the secret ``non-scrambled'' version of the code
denoted $\C_0$ and spanned by the basis:
\begin{equation}\label{eq:C_0}
  \RB_0 = \Biggl(\begin{pmatrix}
    \Av_1 & \Rv_1 \\
    \Rv'_1 & \Rv''_1
\end{pmatrix},\dots,\begin{pmatrix}
    \Av_{km} & \Rv_{km} \\
    \Rv'_{km} & \Rv''_{km}
\end{pmatrix}\Biggr).
\end{equation}

The idea of the combinatorial distinguisher to follow consists in
applying a projection map on both the row and columns spaces of
$\Cpub$ in order to get rid of the contributions of the matrices
$\Rv_i, \Rv_i'$ and $\Rv_i''$ while not destroying the underlying
$\Fqm$--linear structure.

To understand the idea, let us first reason on the non scrambled code
$\C_0$. Choose two matrices
\begin{itemize}
\item $\Uv \in \Fq^{m \times (m+\ell_1)}$ of full rank whose $\ell_1$ last
  columns are $0$, {\em i.e.},
\[
  \Uv = (\Uv_0 ~|~ \mathbf{0}),\qquad \text{with\ \ } \Uv_0 \in \mathbf{GL}_{m}(\Fq).
\]
\item and, for some integer $n'$ such that $k < n' \leq n$, a full rank matrix
  $\Vv \in \Fq^{(n+\ell_2) \times n'}$
  whose last $\ell_2$ rows are $0$:
\[
  \Vv =
  \begin{pmatrix}
    \Vv_0 \\ \mathbf{0}
  \end{pmatrix}
  ,\qquad \text{with\ \ } \Vv_0 \in \Fq^{n\times n'}\ \text{of\
    full\ rank}.
\]
\end{itemize}
Now, observe that the code $\Uv \C_0 \Vv$ is spanned by
\[
\Uv_0 \Av_1 \Vv_0, \dots, \Uv_0 \Av_{km} \Vv_0,
\]
which is a basis of the matrix Gabidulin code
$\Psi_{\gamma'} (\Gabcode{\gv \Vv_0}{m,k,n'})$,
where $\gamma'$ is the image of the basis $\gamma$ by $\Uv_0$.
Hence, this code is distinguishable from random by computing its
left stabilizer algebra.

Note finally that the number of choices of $\Uv, \Vv$ is
$\approx q^{m^2 + nn'}$.  The latter quantity being minimal when
$n' = k+1$ (we should have $n' > k$ since otherwise the resulting
code would be the full matrix space $\Fq^{m\times n'}$).

Now, when considering the public code $\Cpub$ instead of $\C_0$ the very same
observation can be made by replacing $\Uv$ by $\Uv' \eqdef \Uv \Pv^{-1}$ and
$\Vv$ by $\Vv' \eqdef \Qv^{-1}\Vv$ and the number of good choices killing the contributions
of the $\Rv_i, \Rv_i', \Rv_i''$ matrices remains the same, namely:
$q^{m^2+n(k+1)}$.

Thus, the suggested distinguisher consists in repeating the following
operations:
\begin{itemize}
\item Guess the pair $\Uv', \Vv'$ with $\Uv' \in \Fq^{m \times (m+\ell_1)}$
  and $\Vv' \in \Fq^{(n+\ell_2)\times (k+1)}$,
\item Compute the left stabilizer algebra of $\Uv' \Cpub \Vv'$.
\end{itemize}
until you get a stabilizer algebra of dimension $\geq m$.

The probability of finding a valid pair $\Uv', \Vv'$ is
\[
  \mathbb P \approx \frac{q^{m^2 + n(k+1)}}{q^{m(m+\ell_1) + (n+\ell_2)(k+1)}}
  = q^{-(m\ell_1 + (k+1)\ell_2)}
\]
which yields a complexity of \begin{equation}\label{att:str}
\widetilde O ( q^{m\ell_1 + (k+1)\ell_2})
\end{equation}

for the distinguisher.

To conclude on this section let us do some remarks.
\begin{itemize}
\item The process is not symmetric on rows and columns since, on one
  hand $\Uv'$ must kill the random rows while preserving the
  $\Fqm$--linearity. On the other hand, the matrix $\Vv'$ should only
  kill the random columns, even if it partially punctures the
  Gabidulin code.
\item The above distinguisher holds when replacing the Gabidulin code
  by any $\Fqm$--linear code, it does not use the Gabidulin
  structure. In the subsequent section, the proposed algebraic attacks
  will try to take advantage of the Gabidulin structure.
\end{itemize}

\subsubsection{An Overbeck-like distinguisher.}
As already mentioned, the previous distinguisher does not actually
take advantage of the Gabidulin structure and could hold for any
$\Fqm$--linear code. Let us discuss how to take advantage of the
Gabidulin structure. The key of Overbeck's distinsguisher is that,
when given a Gabidulin code represented as a vector code $\Cvec$, one
can easily compute its image by the Frobenius map $\Cvec^{[1]}$ and
the sum $\Cvec + \Cvec^{[1]}$ is small (it has $\Fqm$--dimension
equal to $1+\dim_{\Fqm} \Cvec$) compared to what happens with a random
code. More generally, the sum $\Cvec + \Cvec^{[1]} + \cdots + \Cvec^{[t]}$
is small compared to the random case.

The difficulty in our setting is that we can access neither the
multiplication operation by an element of $\Fqm$ nor the action of
the Frobenius map. We suggest here to identify similar behaviours by
solving a system of bilinear equations. We do not claim that it is the
only manner to distinguish via the resolution of a bilinear
system. However, our observation is that all our attempts led to the
resolution of a quadratic system with $\Theta (m^2)$ unknowns for $\Theta (m^2)$
equations. A linearization would hence lead to $\Theta (m^4)$ unknowns for
only $\Theta (m^2)$ equations. A further analysis of this system or of any
other (and possibly smarter) algebraic modeling would be important to
better assess the security of the system.

In the sequel, we assume that the dimension of the public code
$\Cpub$ satisfies
\[
  2 \dim \Cpub = 2mk \geq \dim \Fq^{(m+\ell_1) \times (n + \ell_2)}.
\]  If this
condition is not satisfied, an algebraic distinguisher of similar
flavor can be searched on the dual code (see further for a discussion on
a structural attack on the dual).

The idea of this distinguisher relies on this
observation.

\begin{lemma}\label{lem:product_matrix_Gabidulin}
  Let $b$ be a non-negative integer.  Let
  $\Mv \in \Psi_{\gamma} (\Gabcode{\gamma}{m,b,m})$ and
  $\Cv \in \Psi_{\gamma}(\Gabcode{\gv}{m,k,n})$. Then,
  \[
    \Mv \Cv \in \Psi_{\gamma}(\Gabcode{\gv}{m, k+b-1, n}).
  \]
\end{lemma}

\begin{proof}
  The matrix $\Mv$ represents a $q$--polynomial $P_{\Mv}$ of $q$--degree
  $< b$ in the basis $\gamma$ (regarding $P_{\Mv}$ as an
  $\Fq$--endomorphism of $\Fqm$). The matrix $\Cv$ represents a
  $q$--polynomial $P_{\Cv}$ of $q$--degree $<k$ from the basis $\gv$
  to the basis $\gamma$.  Thus, $\Mv \Cv$ represents the
  $q$--polynomial $P_{\Mv} \circ P_{\Cv}$ from the basis $\gv$ to the
  basis $\gamma$. The $q$--polynomial $P_{\Mv} \circ P_{\Cv}$ has $q$--degree
  $< k+b-1$. This yields the result. \qed
\end{proof}

\begin{remark}
  The previous lemma is somehow a matrix version of Overbeck's
  distinguisher.
\end{remark}

As for the previous distinguisher, for the sake of clarity, we first
reason on the non scrambled code $\C_0$ whose basis is given in
\eqref{eq:C_0}.

\begin{proposition}\label{prop:dim_product}
  Let $\D \subseteq \Fq^{(m+\ell_1) \times (m+\ell_1)}$ be the matrix code:
  \[
    \D \eqdef \left\{
      \begin{pmatrix}
        \Bv & \mathbf{0} \\ \Tv_1 & \Tv_2
      \end{pmatrix} ~\bigg|~ \Bv \in
      \Psi_{\gamma}(\Gabcode{\gamma}{m,n-k,m}),\ \Tv_1 \in \Fq^{\ell_1
        \times m},\ \Tv_2 \in \Fq^{\ell_1 \times \ell_1}\right\}.
  \]
  Then,
  \begin{enumerate}[(i)]
  \item\label{item:dimD} $\dim_{\Fq} \D = m(n-k) + \ell_1 (m + \ell_1)$;
  \item\label{item:dimU} The code
    $\mathcal U \eqdef \text{Span}_{\Fq} \left\{ \Dv \Cv ~|~ \Dv \in
      \D,\ \Cv \in \C_0 \right\} \subseteq \Fq^{(m+\ell_1) \times
      (n+\ell_2)}$ satisfies:
    \[
      \dim_{\Fq} \mathcal U \leq (m+\ell_1)(n+\ell_2) - m.
    \]
  \end{enumerate}
\end{proposition}

\begin{proof}
  (\ref{item:dimD}) is an immediate consequence of the definition of
  $\D$.  To prove (\ref{item:dimU}), let $\Dv \in \D$ and
  $\Cv \in \C_0 $.
  Then, they have the following shapes
  \[
    \Dv =
    \begin{pmatrix}
      \Bv & \mathbf{0} \\  \Tv_1 & \Tv_2
    \end{pmatrix}
    \quad \text{and} \quad
    \Cv =
    \begin{pmatrix}
      \Av & \Rv \\ \Rv' & \Rv'' 
    \end{pmatrix},
  \]
  where $\Bv \in \Psi_{\gamma}(\Gabcode{\gamma}{m,n-k,m})$ and
  $\Av \in \Psi_{\gamma}(\Gabcode{\gv}{m,k,n})$ and the other matrices
  are arbitrary.  According to
  Lemma~\ref{lem:product_matrix_Gabidulin}, we deduce that
  $\Bv \Av \in \Psi_{\gamma}(\Gabcode{\gv}{m,n-1,n})$.  Hence,
  any element of $\mathcal U$ has the shape
  \[
    \begin{pmatrix}
      \Cv & \Sv \\ \Sv' & \Sv''
    \end{pmatrix},
  \]
  where $\Cv \in \Psi_{\gamma}(\Gabcode{\gv}{m,n-1,n})$ and the
  matrices $\Sv, \Sv'$ and $\Sv''$ are arbitrary. This yields the
  upper bound on the dimension of $\mathcal U$.\qed
\end{proof}

\begin{corollary}\label{cor:dim_prod_Cpub}
  There exists a matrix code
  $\D' \subseteq \Fq^{(m+\ell_1)\times (m+ \ell_1)}$ such
  that
  \[\dim_{\Fq} \text{Span}_{\Fq} \left\{\Dv \Cv ~|~ \Dv \in \D',\ \Cv
      \in \Cpub \right\} \leq (m+\ell_1)(n+\ell_2) - m.\]
\end{corollary}

\begin{proof}
  Recall that $\Cpub = \Pv \C_0 \Qv$. Then, setting
  $\D' \eqdef \Pv \D \Pv^{-1}$ where $\D$ is the code of
  Proposition~\ref{prop:dim_product}, yields the result. \qed
\end{proof}

Corollary~\ref{cor:dim_prod_Cpub} is the key of our forthcoming
distinguisher : it shows that some operation (the ``left
multiplication'') by $\D'$ does not fill in the ambient space while it
would fill it in if $\Cpub$ was random.

However, the difficulties are that:
\begin{enumerate}[(1)]
\item the code $\D'$ is unknown;
\item\label{item:issue2} we need to equate the fact that the dimension of the span of
  $\D' \Cpub$ is not $(m+\ell_1)(n+ \ell_2)$.
\end{enumerate}
To address (\ref{item:issue2}), we proceed as follows. We will define
a formal variable $\Dv \in \Fq^{(m+\ell_1) \times (m+\ell_1)}$. Note
first that the matrix code $\mathcal{D}$ of
Proposition~\ref{prop:dim_product} contains the identity matrix. Then,
Corollary~\ref{cor:dim_prod_Cpub} asserts that $\Cpub + \Dv \Cpub$ is
not equal to the ambient space since it is contained in the code
$\D' \Cpub$ of codimension at least $m$. Hence its dual
$(\Cpub + \Dv \Cpub)^\perp$ is nonzero. Since
$(\Cpub + \Dv \Cpub)^\perp \subseteq \Cpub^\perp$ there exists
$\Mv \in \Cpub^\perp$ satisfying:
\[
  \forall \Cv \in \Cpub,\ \text{Tr}(\Dv \Cv \Mv^t) = 0. 
\]
Any $\Mv \in \Cpub^\perp$ solution of the above system is an element
of $(\Cpub + \Dv \Cpub)^\perp$. Thus, we can set the following bilinear system with
\begin{itemize}
\item {\bf Unknowns:} $\Dv \in \Fq^{(m+\ell_1)\times (m +\ell_1)}$ and
  $\Mv \in \Cpub^\perp$;
\item {\bf Equations:} for any element $\Cv$ of our $\Fq$--basis $\RB$
  (see (\ref{eq:RBpub})),
  \[
    \text{Tr}(\Dv \Cv \Mv^t) = 0.
  \]
\end{itemize}
For this bilinear system we can count the equations and unknowns.
\begin{itemize}
\item {\bf Number of unknowns:}
  \begin{itemize}
  \item $\Dv$ has $(m+\ell_1)^2$ entries but,
  from Proposition~\ref{prop:dim_product}, it lies in a space of dimension
  $m(n-k-1) + \ell_1(m+\ell_1)$. Hence one can specialize some variables and restrict to $m(m- n +\ell_1  + k +1) + 1$ variables;
\item $\Mv$ is in $\Cpub^\perp$ which has dimension
  $(m+\ell_1)(n+\ell_2)-mk$, one can even specialize $m-1$ variables since
  $(\Cpub + \Dv \Cpub)^\perp$ has codimension at least $m$.
\end{itemize}
\item {\bf Number of equations:} it is nothing but $\dim_{\Fq} \Cpub = mk$.
\end{itemize}

Finally, recall that we assumed
$2\dim \Cpub = 2mk \geq (m+\ell_1)(n+\ell_2)$ which leads to the fact
that, $\Cpub + \Dv \Cpub$ would fill in the ambient space if $\Cpub$
was random. Indeed, it is easy to check that for a random matrix code $\Crand$
the dimension of $\Crand + \Dv \Crand$ is typically $\min (m+\ell_1 (n+\ell_2), 2mk) = m+\ell_1 (n+\ell_2)$ or equivalently that $\Crand + \Dv \Crand$ fills in the ambient space.

However, assuming that $k = \Theta (m)$ and $m \approx n$, we have
$\Theta (m^2)$ bilinear equations with $\Theta(m^2)$ unknowns from
$\Dv$ and from $\Mv$. For the parameters we consider, with
$m \approx 40$, this represents thousands of variables to
handle. Thus, we claim that our system remains out of reach of such a
distinguisher.

\subsubsection{Attacking the dual.}
The dual of an enhanced Gabidulin code is
described in Proposition~\ref{prop:dual_emc}:
after applying left and right multiplying by ${(\Pv^t)}$
and ${(\Qv^{t})}$ respectively, the dual of the public code
has the following shape
\[
  \Dmat = \mathcal{U}_{n-k} \oplus \mathcal W, 
\]
where $\mathcal{U}_{n-k}$ is a matrix
version of a Gabidulin code ``extended by
zero'' \emph{i.e.} to which $\ell_1$ rows and $\ell_2$ columns of zero
have been added and $\mathcal W$ is some complement subspace of
dimension $n \ell_1+ m\ell_2 + \ell_1 \ell_2$. We assume here that
$2 \dim_{\Fq} \Cmat \geq (m+\ell_1)(n+\ell_2)$. Equivalently, we suppose
the dual of the public code to have rate $>1/2$

The code $\mathcal{U}_{n-k}$ is derived from a Gabidulin code of $\Fqm$--dimension $n-k$ and hence has $\Fq$--dimension $m(n-k)$.
Similarly to the previous attack, observe that if
$\Dv \in \Fq^{(m+\ell_1)\times (m+\ell_{1})}$ is in some matrix version
of a Gabidulin code of dimension $b$ extended by zero, then
\[
  (\mathcal{U}_{n-k} + \Dv \mathcal{U}_{n-k}) \subseteq \mathcal{U}_{n-k+b-1},
\]
where $\mathcal{U}_{n-k+b-1}$ is a matrix version of a Gabidulin code
of $\Fqm$--dimension $n-k+b-1$ extended by zero.
Thus,
\begin{align*}
  \dim_{\Fq} (\Dmat + \Dv \Dmat) &\leq \dim_{\Fq} (\mathcal{U}_{n-k} + \Dv \mathcal{U}_{n-k}) + \dim_{\Fq} \mathcal W + \dim_{\Fq}  \Dv \mathcal{W} \\
  &\leq m(n-k+b-1) + 2(n \ell_1 + m \ell_2 + \ell_1 \ell_2).
\end{align*}
Under the assumption that $\Dmat$ has rate $>1/2$, the code $\Dmat + \Dv \Dmat$
would equal the ambient space w.h.p. if $\Dmat$ was random, while, it actually
does not as soon as:
\[
  m(n-k+b-1) + 2(n\ell_1 + m\ell_2 + \ell_1 \ell_2) < (m+\ell_1)(n+\ell_2),
\]
which holds as soon as
\begin{equation}\label{eq:good_b}
  b < k+1 - \frac{n}{m}\ell_1 - \ell_2 - \frac{\ell_1 \ell_2}{m}\cdot
\end{equation}
Here again as in the previous case, we can equate this distinguisher
by choosing a $b$ satisfying (\ref{eq:good_b}),  searching
$\Dv \in \Fq^{(m + \ell_1)\times (m+\ell_1)}$ and $\Mv \in \Cmat$ such
that for all $\Cv$ in the dual of the public code, we have
\(\text{Tr}(\Dv \Cv \Mv^t) = 0\).

As in the previous case, the number of unknowns remains prohibitive
compared to the number of equations.

\begin{remark}
  Solving such a system for $b = 0$ corresponds to searching the hidden
  $\Fqm$--linearity of the code $\mathcal{U}_{n-k}$. 
\end{remark}

\begin{remark}
  We also considered a combinatorial approach on the dual as we
  proposed on the public code itself. The attack on the public code
  consists in guessing a relevant puncturing of the row and column
  spaces in order to get rid of the contribution of the random rows
  and columns added in the enhancing process. When reasoning on the
  dual, we need to get rid of the contribution of the complement
  subspace $\mathcal W$. This can be done by guessing a relevant
  shortening of the row and column space. Such a structural attack on
  the dual turns out to be equivalent to the previously presented
  combinatorial attack performed on the public code itself.
\end{remark}

\subsubsection{Finding codewords just below the Singleton bound and exploiting their structure.}
In this subsection we consider the approach of finding matrices in the public code of weight the MRD bound minus 1. These matrices could be candidates for unraveling the structure of the hidden Gabidulin code. We first consider how the MRD bound is affected by the enhanced matrix codes transformation (i.e. adding random rows and columns). Then, we show there are many matrices of weight the MRD bound minus 1 in the public code and that their structure does not seem to reveal any information on the secret key.

 To simplify the analysis, we only consider the case $n = m$ since all parameters in Section \ref{sec:param} satisfy this equality.
Without loss of generality, since we can always consider the transposed code (\emph{i.e.} the code $\C_{mat}^t = \{ \Mv^t ~|~ \Mv \in \C_{mat} \}$, which is different from the dual code), we assume in this subsection $\ell_1 \ge \ell_2 > 0$.
\begin{definition}[Singleton bound \cite{D78}]
A matrix code $\C_{mat}[m \times n, K, d]$ satisfies the Singleton (or MRD) bound if:
\[ d \le min(n,m) - \frac{K}{max(n,m)} + 1.\]
Codes achieving the equality for this bound are called MRD codes.
\end{definition}

\begin{lemma}
The minimum distance for the code $\C_{pub}[(m+\ell_1)\times(m+\ell_2), km, d]$ satisfies:
\[ d \le m - k + 1 + \ell_2 + \floor{\frac{k\ell_1}{m+\ell_1}}. \]
\end{lemma}

\begin{proof}
The direct application of the Singleton bound yields:
\[ d \le m - k + 1 + \ell_2 + \frac{k\ell_1}{m+\ell_1}\cdot\]
Since $d$ is an integer, the integer part of $\frac{k\ell_1}{m+\ell_1}$ can be taken. \qed
\end{proof}

This shows that after perturbation, the Singleton bound is increased by at least $\ell_2$.

Let us denote $d_0 := m - k + 1 + \ell_2  + \floor{\frac{k\ell_1}{m+\ell_1}}$. It is legitimate to wonder whether $\C_{pub}$ achieves the MRD bound and in the contrary, how many matrices of rank $d_0 - 1$ are contained in $\C_{pub}$.

In the following we prove a lower bound on the expected number of matrices in $\C_{pub}$ of rank $d_0 - 1$, which shows that $\C_{pub}$ is highly unlikely to be an MRD code and that, moreover, it contains lots of matrices of rank $d_0 - 1$.

\begin{lemma}
Let $N$ be the expected number of matrices in $\C_{pub}$ of rank $d_0 - 1$. It achieves the following inequality
\[ N \ge \left[ \begin{array}{c}
	{m} \\ {m-k+1} \end{array}\right]_q(q^m - 1) \times q^{(\ell_1+1)(1-k) - 1} \approx q^{m+(m-k-\ell_1)(k-1) - 1}. \]
\end{lemma}

\begin{proof}
The matrices in $\C_{pub}$ are of the form
\[ \Mv = \Pv \left( \begin{array}{c|c}
	\Gv & \multirow{2}{*}{\Bv}\\
	\Av & 
\end{array}\right) \Qv \]
with $\Gv \in \Gabcode{\gv}{m,k,m}$ of size $m \times m$, $\Av$ of size $\ell_1 \times m$ and $\Bv$ of size $(m+\ell_1) \times \ell_2$.

Sufficient conditions such that $\Mv$ is of rank $d_0 - 1$ are:
\begin{enumerate}
	\item $\Gv$ is of rank $d_0 - \ell_2$
	\item $RowSpace(\Av) \subset RowSpace(\Gv)$
	\item $\Bv$ is of full rank $\ell_2$
	\item $\exists i \in [1,\ell_2], \Bv_{*,i} \in ColSpace\begin{pmatrix}\Gv \\ \Av\end{pmatrix}$
	\item $ColSpace(\Bv_{*,i})_{\substack{j \in [1,\ell_2] \\ j \neq i}} \cap ColSpace\begin{pmatrix}\Gv \\ \Av\end{pmatrix} = \{ 0\}$
\end{enumerate}

The number of matrices $\Gv$ of rank $m-k+1$ is $\left[ \begin{array}{c}
{m} \\ {m-k+1} \end{array}\right]_q(q^m - 1)$ \cite{L07}. Because $d_0 - \ell_2 \ge m-k+1$ and since the weight distribution in a Gabidulin code is increasing with the weight~\cite{bartz2022rank}, the number of matrices $\Gv$ of rank $d_0 - \ell_2$ is $\ge \left[ \begin{array}{c} m \\ m-k+1 \end{array} \right]_q(q^m - 1)$. It remains now to evaluate the probabilities of 2.-5. We assume that $\Av$ and $\Bv$ are uniformly random independent matrices.

The probability that a single row of $\Av$ is contained in $RowSpace(\Gv)$, which is a subspace of $\Fq^m$ of dimension $m-k+1$ is $q^{1-k}$. Hence the probability of 2. is $q^{\ell_1(1-k)}$.

The probability of 3. is $\prod_{j = 0}^{\ell_2 - 1} (1 - q^{j-\ell_2-m}) \ge (1-q^m)^{\ell_2} \ge q^{-1}$.

Since $\Bv$ is of full rank, $\dim ColSpace(\Bv_{*,i})_{\substack{j \in [1,\ell_2] \\ j \neq i}} =  \ell_2 - 1$. Noting that\break $ \dim ColSpace\begin{pmatrix}\Gv \\ \Av\end{pmatrix} = m-k+1$, the probability of 4. and 5. together is
\[ \ell_2 q^{(1-k)} (1 - q^{(\ell_2-1)(1-k)}) \ge q^{1-k}. \] \qed
\end{proof}

The above lemma shows that for the parameters presented in Section \ref{sec:param}, the expected number of matrices of rank $d_0 - 1$ is very large (at least $q^{128}$ for all parameters presented in Figure \ref{MG-main}) and does not seem to yield any information to retrieve the secret Gabidulin structure.

The above proof can be easily adapted in the special case of $\ell_2 = 0$ to find a similar upper bound as long as $\frac{k\ell_1}{m+\ell_1} \ge 1$, which is always the case in our parameters.

\subsection{Attacks on the message}

In all encryption schemes derived from the MinRank-McEliece and Niederreiter encryption frames, an opponent who wants to retrieve the original message without knowing the secret key has to solve a generic  $\MR (q,m,n,k,r)$ instance. We recall here the main attacks on this problem.

\subsubsection{Hybrid approach.}

In order to improve the attacks on MinRank, a generic approach has been introduced in \cite{BBBGT22}, which consists in solving smaller instances. The complexity is given by: \begin{equation}\label{att:hyb}
\min_a q^{ar}\mathcal{A}(q,m,n-a,K-am,r)
\end{equation} where $\mathcal{A}$ is the cost of an algorithm to solve a MinRank instance.

\subsubsection{The kernel attack.}

This attack was described in \cite{GC00}. The idea of the attack consists in sampling random vectors, hoping that they are in the kernel of $\Ev$, and deducing a linear system of equations. Its complexity is equal to:\begin{equation}\label{att:comb}
O(q^{r\lceil \frac{k}{m} \rceil}k^\omega).
\end{equation} 

\subsubsection{Minors attack.}

This algebraic attack was introduced and studied in \cite{FSS10}. This modeling uses the minors of $\Ev$. This method has been improved in \cite{GND23}. We refer to these papers for the complexity of the attack.

\subsubsection{Support Minors attack.}

The Support Minors modeling was introduced in \cite{BBCGPSTV20}. This idea also uses minors of a matrix, giving us an other system of equations. With this approach, the complexity is of: \begin{equation}\label{att:alg}
O\left(N_bM_b^{\omega -1}\right)
\end{equation} where \begin{align*}
N_b &= \sum_{i=1}^b (-1)^{i+1} \binom{n}{r+i} \binom{k+b-i-1}{b-i} \binom{k+i-1}{i}\\
M_b &= \binom{n}{r}\binom{k+b-1}{b}
\end{align*} and $b$ is the degree to which we augment the Macaulay matrix of the system.


\section{Parameters}\label{sec:param}

\subsection{Parameters for EGMC-Niederreiter encryption scheme}

We apply our idea of Enhanced MinRank-Niederreiter encryption frame by taking a Gabidulin code $\mathcal{G}[m,k]_{q^m}$, for which we know an efficient decoding algorithm, allowing to decode errors of weight up to $\lfloor \frac{m-k}{2} \rfloor$. As we previously said, decrypting the ciphertext $\Yv$ without the secret key $\sk$ is strictly equivalent to the MinRank problem of parameters $(q,m+\ell_1,m+\ell_2,km,r)$. 

\subsubsection*{Choice of parameters.}

For a given value of $r$, we choose jointly the values of $m$ and $k$ (the optimal parameters verify $r=\lfloor \frac{m-k}{2} \rfloor$) such that the underlying MinRank instance is secure against the known attacks. Then, we choose $\ell_1$ and $\ell_2$ to obtain parameters resistant to structural attacks (see below).

We choose the parameters based on the best known attacks, see Figure \ref{MG-main} and Figure \ref{MG-other}. Among these attacks, three of them rely on solving the associated MinRank instance: Alg. refers to the algebraic Support Minors attack (see Equation \ref{att:alg}), Hyb. refers to the Hybrid approach (see Equation \ref{att:hyb}), and Comb. refers to the combinatorial Kernel attack (see Equation \ref{att:comb}). We also consider the attack which consists in retrieving the $\Fqm$-linear structure of the code $\Cmat$ (see Equation \ref{att:str}). 

We propose two kinds of parameters. In our main parameters, we add as many rows as columns to the matrix Gabidulin code ($\ell_1=\ell_2$). The resulting parameters for 128 bits, 192 bits and 256 bits of security can be found in Figure \ref{MG-main}. We propose also some sets of parameters for which only rows or only columns are added to the matrix Gabidulin code ($\ell_1=0$ or $\ell_2=0$), which can be found in Figure \ref{MG-other}. We consider NIST-compliant parameters, while maintaining a margin of 15 bits above the security level.

The public key $\pk$ consists in a parity check matrix of $\Cmat'$. As a linear matrix code $[(m+\ell_1)(m+\ell_2),km]_q$, it can be represented with: $$km ((m+\ell_1)(m+\ell_2)-km)\log_2 q$$ bits. The number of bits which compose the ciphertext $\cv\in\Fqm^{(m+\ell_1)(m+\ell_2)-km}$ is equal to: $$((m+\ell_1)(m+\ell_2)-km)\log_2 q.$$

\begin{figure}[h]
\begin{center}
{\setlength{\tabcolsep}{0.4em}
{\renewcommand{\arraystretch}{1.6}
{\scriptsize
  \begin{tabular}{|c||c|c|c|c|c|c||c|c|c|c||c|c|}
    \hline
    Sec. & $q$ & $k$ & $m$ & $\ell_1$ & $\ell_2$ & $r$ & Alg. & Hyb. & Comb. & Struc. & $\pk$ & $\ct$ \\ \hline\hline
    \multirow{4}{*}{128}
    & 2 & 17 & 37 & 3 & 3 & 10 & 193 & 170 & 179 & 165 & 76 kB & 121 B \\ \cline{2-13} 
    & 2 & 25 & 37 & 3 & 3 & 6 & 168 & 150 & 164 & 189 & 78 kB & 84 B \\ \cline{2-13} 
    & 2 & 35 & 43 & 2 & 2 & 4 & 158 & 145 & 158 & 158  & 98 kB & 65 B \\ \cline{2-13}
    & 2 & 47 & 53 & 2 & 2 & 3 & 158 & 147 & 161 & 202  & 166 kB & 66 B \\  \hline\hline 
    192 & 2 & 51 & 59 & 2 & 2 & 4 & 222 & 209 & 224 & 222 & 268 kB & 89 B \\ \hline\hline
   \multirow{2}{*}{256}
    & 2 & 23 & 47 & 3 & 3 & 12 & 302 & 271 & 285 & 284 & 191 kB & 177 B \\ \cline{2-13} 
   & 2 & 37 & 53 & 3 & 2 & 8 & 315 & 290 & 310 & 273 &  274 kB & 139 B \\ \cline{2-13} 
   & 2 & 71 & 79 & 2 & 2 & 4 & 303 & 289 & 305 & 302 & 667 kB & 119 B \\ \hline
  \end{tabular}
  \vspace{0.5\baselineskip}
  \caption{Reference parameters for the EGMC-Niederreiter encryption scheme}
  \label{MG-main}
}}}
\end{center}
\end{figure}

\begin{figure}[h]
\begin{center}
{\setlength{\tabcolsep}{0.4em}
{\renewcommand{\arraystretch}{1.6}
{\scriptsize
  \begin{tabular}{|c||c|c|c|c|c|c||c|c|c|c||c|c|}
    \hline
    Sec. & $q$ & $k$ & $m$ & $\ell_1$ & $\ell_2$ & $r$ & Alg. & Hyb. & Comb. & Struc. & $\pk$ & $\ct$ \\ \hline\hline
    \multirow{3}{*}{128}
    & 2 & 17 & 37 & 4 & 0 & 10 & 181 & 168 & 179 & 148 & 70 kB & 111 B \\ \cline{2-13}
    & 16 & 13 & 23 & 1 & 1 & 5 & 236 & 273 & 282 & 148 & 41 kB & 138 B \\ \cline{2-13}
    & 16 & 7 & 23 & 0 & 5 & 8 & 172 & 262 & 276 & 160 & 33 kB & 207 B \\ \hline\hline
    \multirow{3}{*}{192}
    & 2 & 23 & 43 & 5 & 0 & 10 & 239 & 220 & 230 & 215 & 133 kB & 134 B \\ \cline{2-13}  
    & 2 & 33 & 47 & 5 & 0 & 7 & 238 & 221 & 232 & 235 & 173 kB & 111 B \\ \cline{2-13} 
    & 2 & 41 & 53 & 4 & 0 & 6 & 258 & 240 & 257 & 212 & 230 kB & 106 B \\ \hline\hline 
    \multirow{2}{*}{256}
   & 16 & 9 & 29 & 2 & 1 & 10 & 310 & 373 & 382 & 272 &  87 kB & 334 B \\ \cline{2-13} 
   & 16 & 17 & 29 & 2 & 1 & 8 & 357 & 399 & 408 & 304 &  107 kB & 218 B \\ \hline
  \end{tabular}
  \vspace{0.5\baselineskip}
  \caption{Alternative parameters in particular case of $\ell_1=0$ or $\ell_2=0$, or $q>2$}
  \label{MG-other}
}}}
\end{center}
\end{figure}

\subsection{Comparison with other schemes}

We propose a comparison of our sizes with those of other encryption schemes based on various problems, see Figures \ref{size-comp128}, \ref{size-comp192} and \ref{size-comp256}. Note that we achieve better performances than RQC and ROLLO, which are other rank-based encryption schemes, and even the smallest ciphertext sizes compared to the schemes proposed to the NIST based on codes and lattices.

\begin{figure}[h]
    \begin{minipage}{0.45\linewidth}
        \centering
        \setlength{\tabcolsep}{0.3em}
        \renewcommand{\arraystretch}{1.6}
        \scriptsize
        \begin{tabular}{|c||c|c|}
    \hline
    Scheme & $\pk$ & $\ct$ \\ \hline\hline
    \textbf{EGMC-Niederreiter}, Fig. \ref{sch-nied} & 98 kB & 65 B \\ \hline
    Classic McEliece \cite{BCLMNPPSSSW17} & 261 kB & 96 B \\ \hline 
   \textbf{EGMC-Niederreiter}, Fig. \ref{sch-nied} & 33 kB & 207 B \\ \hline
   ROLLO I \cite{ABDGHRTZABBBO19} & 696 B & 696 B \\ \hline
   KYBER \cite{Kyber} & 800 B & 768 B \\ \hline
   RQC-Block-NH-MS-AG \cite{ABDGV23} & 312 B & 1118 B \\ \hline
   BIKE \cite{AABBBBDGGGMMPSTZ17} & 1540 B & 1572 B \\ \hline
   RQC-NH-MS-AG \cite{BBBG22} & 422 B & 2288 B \\ \hline
   RQC \cite{AABBBDGZ17} & 1834 B & 3652 B \\ \hline
   HQC \cite{AABBBDGPZ21a} & 2249 B & 4481 B \\ \hline
  \end{tabular}
        \vspace{0.5\baselineskip}
        \caption{Comparison of different schemes for 128 bits of security}
        \label{size-comp128}
    \end{minipage}
    \hfill
    \begin{minipage}{0.45\linewidth}
        \centering
        \setlength{\tabcolsep}{0.3em}
        \renewcommand{\arraystretch}{1.6}
        \scriptsize
        \begin{tabular}{|c||c|c|}
    \hline
    Scheme & $\pk$ & $\ct$ \\ \hline\hline
   \textbf{EGMC-Niederreiter}, Fig. \ref{sch-nied} & 268 kB & 89 B \\ \hline
   \textbf{EGMC-Niederreiter}, Fig. \ref{sch-nied} & 133 kB & 134 B \\ \hline
   Classic McEliece \cite{BCLMNPPSSSW17} & 524 kB & 156 B \\ \hline
   ROLLO I \cite{ABDGHRTZABBBO19} & 958 B & 958 B \\ \hline
   KYBER \cite{Kyber} & 1184 B  & 1088 B  \\ \hline
   RQC-Block-NH-MS-AG \cite{ABDGV23} & 618 B & 2278 B \\ \hline
   BIKE \cite{AABBBBDGGGMMPSTZ17} & 3082 B  & 3024 B \\ \hline
   RQC-NH-MS-AG \cite{BBBG22} & 979 B & 3753 B \\ \hline
   RQC \cite{AABBBDGZ17} & 2853 B & 5690 B \\ \hline
   HQC \cite{AABBBDGPZ21a} &  4522 B & 9026 B  \\ \hline
   
  \end{tabular}
        \vspace{0.5\baselineskip}
        \caption{Comparison of different schemes for 192 bits of security}
        \label{size-comp192}
    \end{minipage}
\end{figure}

\begin{figure}[h]
\begin{center}
{\setlength{\tabcolsep}{0.3em}
{\renewcommand{\arraystretch}{1.6}
{\scriptsize
\begin{tabular}{|c||c|c|}
   \hline
    Scheme & $\pk$ & $\ct$ \\ \hline\hline
   \textbf{EGMC-Niederreiter}, Fig. \ref{sch-nied} & 667 kB & 119 B \\ \hline
   Classic McEliece \cite{BCLMNPPSSSW17} & 1044 kB  & 208 B  \\ \hline
   \textbf{EGMC-Niederreiter}, Fig. \ref{sch-nied} & 87 kB & 334 B \\ \hline
   ROLLO I \cite{ABDGHRTZABBBO19} & 1371 B & 1371 B \\ \hline
   KYBER \cite{Kyber} & 1568 B  & 1568 B \\ \hline
   BIKE \cite{AABBBBDGGGMMPSTZ17}  & 5121 B & 5153 B \\ \hline
   RQC \cite{AABBBDGZ17} & 4090 B & 8164 B \\ \hline
   HQC \cite{AABBBDGPZ21a} & 7245 B & 14465 B \\ \hline
   
  \end{tabular}
  \vspace{0.5\baselineskip}
  }}}
\end{center}\caption{Comparison of different schemes for 256 bits of security}
\label{size-comp256}
\end{figure}

\section{Conclusion and further work}

This work presents a general McEliece-like encryption frame for matrix codes whose security is based on the MinRank problem. We propose a general masking for rank codes that we apply on a matrix code version of Gabidulin vector code. We define a new problem: the Enhanced Gabidulin Matrix Code (EGMC) distinguishing problem, for which we propose a thorough analysis and study possible distinguishers to solve it. It results in a competitive encryption scheme, which achieves a ciphertext size of 65B for 128 bits of security.

For future work, it would be interesting to extend our trapdoor by also considering subcodes of matrix codes that we obtain with our masking. Such modifications would probably make the action of the distinguisher more complex, and thus may permit to reduce the parameters of our scheme.




\clearpage

\bibliographystyle{plain}
\bibliography{biblio,codecrypto}

\end{document}